\documentclass[a4paper,UKenglish, cleveref, autoref, thm-restate]{lipics-v2021}

\usepackage{tikz}
\usetikzlibrary{arrows.meta,calc,decorations.markings}
\usepackage{amsmath,amssymb,amsthm}
\usepackage[dvipsnames]{xcolor}
\usepackage{graphicx,xspace}

\usepackage{mathabx}
\usepackage{cite}

\newcommand{\Reals}{\mathbb{R}}

\newcommand{\area}{\mathrm{area}}

\newcommand{\diam}{\mathrm{diam}}

\newcommand{\dist}{\mathrm{dist}}

\DeclareMathOperator{\poly}{poly}
\newcommand{\etal}{\emph{et al.}\xspace}
\newcommand{\eps}{\varepsilon}
\newcommand{\cF}{\mathcal{F}}
\newcommand{\cB}{\mathcal{B}}
\newcommand{\cO}{\mathcal{O}}
\newcommand{\bn}{\mathbf{n}}
\newcommand{\sph}{\mathbb{S}}

\newtheorem{question}[theorem]{Question}

\newenvironment{myquote}%
  {\list{}{\leftmargin=4mm\rightmargin=4mm}\item[]}%
  {\endlist}

\newcommand{\tsp}{{\sc tsp}\xspace}
\newcommand{\mtsp}{{\sc Metric tsp}\xspace}
\newcommand{\etsp}{{\sc Euclidean tsp}\xspace}
\newcommand{\otsp}{{\sc tsp with obstacles}\xspace}
\newcommand{\poltsp}{{\sc tsp in a simple polygon}\xspace}

\newcounter{ctr}
\loop
 \stepcounter{ctr}
 \expandafter\edef\csname c\Alph{ctr}\endcsname{\noexpand\mathcal{\Alph{ctr}}}
\ifnum\thectr<26
\repeat

\newif\ifComments
\Commentstrue
\ifComments
   \newcommand{\skb}[1]{\textcolor{red}{S: #1}}
   \newcommand{\lt}[1]{\textcolor{orange}{ LT: #1}}
\else
    \newcommand{\skb}[1]{}
    \newcommand{\lt}[1]{}
\fi

\usepackage{microtype} 

\title{Realizing Metric Spaces with Convex Obstacles}

\author{S\'andor Kisfaludi-Bak \footnote{This work was supported by the Research Council of Finland, Grant 363444.}}{Aalto University, Finland}{sandor.kisfaludi-bak@aalto.fi}{https://orcid.org/0000-0002-6856-2902}{}
\author{Leonidas Theocharous}{University of Ottawa, Canada}{ltheocha@uottawa.ca}{https://orcid.org/0000-0002-1707-6787}{}

\authorrunning{S. Kisfaludi-Bak, L. Theocharous}

\Copyright{S\'andor Kisfaludi-Bak, Leonidas Theocharous}



\category{} 

\relatedversion{} 

\hideLIPIcs
\acknowledgements{}

\begin{document}
\nolinenumbers

\maketitle

\begin{abstract} 
 The presence of obstacles has a significant impact on distance computation, motion-planning, and visibility. These problems have been studied extensively in the planar setting, while our understanding of these problems in $3$- and higher-dimensional spaces is still rudimentary. In this paper, we study the impact of different types of obstacles on the induced geodesic metric in $3$-dimensional Euclidean space. We say that a finite metric space $(X, \dist_X)$ is \emph{approximately realizable} by a collection $\mathcal{T}$ of obstacles in $\mathbb{R}^3$ if for any $\eps>0$ it can be embedded into $(\mathbb{R}^3\setminus \bigcup_{T\in\mathcal{T}} T,\dist_\mathcal{T})$ with worst-case multiplicative distortion $1+\eps$, where $\dist_\mathcal{T}$ denotes the geodesic distance in the free space induced by $\mathcal{T}$. We focus on three key geometric properties of obstacles --convexity, disjointness, and fatness-- and examine how dropping each one of them affects the existence of such embeddings.

Our main result concerns dropping the fatness property: we demonstrate that any finite metric space is realizable with $1+\varepsilon$ worst-case multiplicative distortion using a collection of convex and pairwise disjoint obstacles in $\mathbb{R}^3$, even if the obstacles are congruent and equilateral triangles. Based on the same construction, we can also show that if we require fatness but drop any of the other two properties instead, then we can still approximately realize any finite metric space.

Our results have important implications on the approximability of \otsp, a natural variant of \tsp introduced recently by Alkema~\emph{et al.}~(ESA~2022). Specifically, we use the recent results of Banerjee \emph{et al.} on \tsp in doubling spaces~(FOCS~2024) and of Chew~\emph{et al.} on distances among obstacles~(Inf. Process. Lett. 2002) to show that \otsp admits a PTAS if the obstacles are convex, fat, and pairwise disjoint. If any of these three properties is dropped, then our results, combined with the APX-hardness of \mtsp, demonstrate that \otsp is APX-hard.
\end{abstract}

\newpage
\section{Introduction}

Understanding navigation~\cite{schwartz83,doi:10.1177/027836498300200304,10.5555/49142,Halperin2018MotionPlanning,agarwal24}, distance ~\cite{wangdahl1974,wang21,asano1986,sharir84,hersh99,mitchell93} and visibility~\cite{Chvatal1975Combinatorial,orourke1987art,fekete2014chromatic,rieck_et_al:LIPIcs.ISAAC.2022.67,doi:10.1137/1.9781611973105.60,lubiw_et_al:LIPIcs.ESA.2021.65} problems among obstacles has been one of the well-motivated research directions in computational geometry. Most of the work has focused on the planar setting, where the metric could be the geodesic distance inside a simple polygon, or more generally in a polygon with holes in the Euclidean plane.
In the plane one can benefit from the fact that paths will often cross, whereas in higher dimensions paths will typically stay disjoint. The additional degree of freedom in $3$- and higher-dimensional space makes the fundamental distance and navigation problems among obstacles harder, and they have not been explored as thoroughly.

Fortunately, we can still get a useful metric structure for certain obstacles in $d$-dimensional Euclidean space (henceforth denoted by $\Reals^d$). We say that a set $\cT$ of obstacles is $\alpha$-\emph{fat} if for each obstacle object $T\in \cT$ we have $\frac{r_{in}(T)}{r_{circum}(T)}\geq \alpha$, where $r_{in}(T)$ denotes the radius of the maximum inscribed ball of $T$, and $r_{circum}(T)$ denotes the radius of the minimum enclosing ball of $T$. We will think of $\alpha$ as a universal constant and talk about collections of fat objects. 
For a set $\cT$ of obstacles let $\cF_\cT=\Reals^d\setminus \bigcup_{T\in \cT} T$ denote the \emph{free space} outside the obstacles. We denote by $\dist_\mathcal{T}(u,v)$ the geodesic distance between $u,v\in \cF_\cT$, which is the infimum of the length of piecewise linear curves in $\cF_\cT$ containing $u$ and $v$. Chew \etal ~\cite{chew2002} established that for any set of convex, disjoint and fat obstacles in $\Reals^d$ one can navigate around the obstacles with constant overhead compared to the Euclidean distance. Their theorem implies that $(\cF_\cT,\dist_\cT)$ has bounded \emph{doubling dimension}, that is, for any $r>0$ and any $x\in \cF_\cT$ the geodesic ball of center $x$ and radius $r$ can be covered by at most $2^\delta$ geodesic balls of radius $r/2$, for some constant $\delta$. The minimum number $\delta$ satisfying the above condition is the doubling dimension. In \Cref{sec:doubling} we establish an explicit bound on the doubling dimension among pairwise disjoint convex $\alpha$-fat obstacles in $\Reals^d$. Metrics of bounded doubling dimension are well-studied generalizations of Euclidean spaces, and retain many useful properties of $\Reals^d$.

A natural question is whether the conditions of fatness, convexity, and pairwise disjointness are all needed for the doubling dimension bound. More generally, one wonders what finite metric spaces can be realized if the condition of fatness, disjointness or convexity is dropped. 
We say that a metric space $(X,\dist_X)$, is \emph{realizable} with a set $\cT$ of obstacles in $\Reals^d$ if one can embed $(X,\dist_X)$ isometrically into $(\cF_\cT,\dist_\cT)$. An \emph{approximate realization} allows for an embedding with a worst-case multiplicative distortion of $(1+\eps)$.

The starting point of our investigation is the following question. Let $\mathfrak P$ denote the disjunction of some subset of properties among (i) convex (ii) pairwise disjoint (iii) fat. 

\begin{question}\label{q:metricembed}
    Is it possible to (approximately) realize the finite metric space $(X,\dist_X)$ in~$\Reals^3$ with obstacles of property~$\mathfrak P$?
\end{question}

It is also worth exploring how various types of obstacles impact the algorithms and complexity of some classic distance problems. The study of the shortest path problem among a set of polygonal obstacles $\cP$ goes back at least 50 years, to the work of Wangdahl \etal \cite{wangdahl1974} in 1974. Their approach to compute a shortest path between two points $s,t$ (as well as many later approaches; see for instance \cite{sharir84,asano1986}) relied on running Dijkstra's algorithm on the visibility graph of the vertices of $\cP$ together with $s$ and $t$. Any such method has an $\Omega(n^2)$ lower bound (where $n$ is the number of vertices in $\cP)$, due to the size of the visibility graph. To achieve subquadratic running times, later algorithms shifted to the continuous Dijkstra paradigm \cite{mitchell93,hersh99}. Finally, Wang  presented an optimal $O(n\log{n})$-time, $O(n)$-space algorithm \cite{wang21}. Data structures for two-point shortest path queries have also been studied and, recently, De Berg \etal presented an $O(n^{10+\eps})$-space data structure with $O(\log{n})$ query time \cite{deberg24}. The shortest path problem among obstacles can be generalized to account for movement through regions with different weights, a variation known as the \emph{weighted region problem}. The general problem admits a $(1+\eps)$-approximation algorithm~\cite{papad91}, while cases with restricted region shapes or weights have also been considered~\cite{bose22,mitchell_wrp}.
  While these are all positive results, computing shortest paths among $3$-dimensional obstacles is a much harder problem. In fact, it is NP-hard~\cite{MitchellS04}. On the other hand, the problem admits an FPTAS ~\cite{Clarkson87,Har-Peled99}.

It is also natural to consider classic geometric optimization problems in the presence of obstacles.
Recently, Alkema \etal \cite{alkema23} introduced a natural generalization of \etsp called \otsp. In \otsp, we are given a set of $n$ points to visit and a collection $\cT$ of obstacles of total complexity $m$ to avoid. Our goal is to find the shortest tour between a set of sites in the plane or some higher-dimensional space, while avoiding a given set of obstacles. In~\cite{alkema23} the authors studied a variant of \otsp called \poltsp, where the sites lie in a simple polygon $P$, that is, the obstacle that the salesman has to avoid is the complement of $P$. They gave an exact algorithm for this problem, running in $2^{O(\sqrt{n}\log{n})}+\poly(n,m)$ time.

In this paper, we will demonstrate the algorithmic impact of various types of obstacles via studying approximation algorithms and APX-hardness for \otsp. In particular, we consider the problem in $\Reals^3$ for different kinds of obstacles. Again, let $\mathfrak P$ denote the disjunction of some subset of properties among (i) convex (ii) pairwise disjoint (iii) fat.

\begin{question}\label{q:TSP}
    Is there a PTAS for \otsp in $\Reals^3$ if the obstacles satisfy property~$\mathfrak P$?
\end{question}


\subsection{Our contribution}

Our main contribution is the following theorem, which shows that all finite metric spaces are approximately realizable with pairwise disjoint convex obstacles in $\Reals^3$.

\begin{restatable}{theorem}{thmMain}\label{thm:approxrealizemetric}
    Let $\eps\in (0,1)$ and let $(X,\dist_X)$ be a metric space of size $|X|=n$ and spread $\Phi:=\max_{a,b,c,d\in X}(\frac{\dist_X(a,b)}{\dist_X(c,d)})$. Then there exists a collection $\cT$ of $O(n^{17}\Phi^5/\eps^{5})$ pairwise disjoint congruent equilateral triangular obstacles and an injection $f:X\rightarrow \Reals^3$ such that
    \[\dist_X(a,b) \leq \dist_{\cT}(f(a),f(b))\leq  (1+\eps)\cdot \dist_X(a,b).\]
    for all $a,b\in X$.
    The set of obstacles can be constructed in $\poly(n,\Phi,1/\eps)$ time.
\end{restatable}


Let $\mathfrak P$ denote the disjunction of some subset of properties among (i) convex (ii) pairwise disjoint (iii) fat.
The theorem directly answers \Cref{q:metricembed} if the property $\mathfrak P$ does not require that objects are fat. We can use some ideas from this construction to give a complete answer to \Cref{q:metricembed} in case of approximate realizations: we can approximately realize any metric space with fat pairwise disjoint obstacles 
(that are non necessarily convex) and with convex fat obstacles that are not necessarily disjoint, see \Cref{thm:fatobstacles}.

We note that our construction is quite generic and we get the same result with other convex shapes that are constant-fat in $\Reals^2$, such as with unit disks or squares instead of congruent equilateral triangles. Using \Cref{thm:approxrealizemetric} and \Cref{thm:fatobstacles} we can use TSP lower bounds designed for general metric spaces to answer \cref{q:TSP}.

\begin{restatable}{corollary}{corAlgoLower}
\label{cor:algoandlower}
    Let $\cT$ be a collection of obstacles in $\Reals^3$ with property $\mathfrak P$ of total complexity $m=\poly(|X|)$.
    Then \otsp has a PTAS if  $\mathfrak P=$ ``convex and pairwise disjoint and fat'', and it is APX-hard for all other~$\mathfrak P$.
\end{restatable}

The PTAS is based on the fact that for $\mathfrak P= \text{``convex and  pairwise disjoint and fat''}$ the free space outside the obstacles has bounded doubling dimension, see \Cref{sec:doubling}. We can also efficiently compute approximate shortest paths using Har-Peled's algorithm~\cite{Har-Peled99}, and running any of the known PTASes for TSP in doubling metrics\cite{BartalGK16,banerjee24} we get the desired algorithm. The lower bound is based on realizing the metric space from the lower bound of Karpinski, Lampis and Schmied~\cite{karpinski} using obstacles. The construction in~\cite{karpinski} has polynomial spread, which allows us to apply \Cref{thm:approxrealizemetric}. See \Cref{sec:algoandlower} for a proof of \Cref{cor:algoandlower}.


\subsection{Related work}

Motion planning is a central problem in robotics, where the goal is to compute a collision-free path for a set of moving objects (robots) in a space with obstacles. The complexity of the problem depends significantly on the underlying environment and the degrees of freedom of the robots. The pioneering works of Schwartz and Sharir \cite{schwartz83,doi:10.1177/027836498300200304}, and later Canny \cite{10.5555/49142}, introduced the first general techniques for solving this problem in the plane. These algorithms have a polynomial dependency on $n$, the combinatorial complexity of the obstacles, but an exponential dependency on $k$, the degrees of freedom of the robots. For a detailed discussion of these foundational works, refer to \cite{Halperin2018Robotics,Halperin2018MotionPlanning}. To highlight the problem's inherent difficulty, note that even in the restricted case of rectangular robots navigating a rectangular environment, the problem is PSPACE-hard \cite{Hopcroft84}. Recently, Agarwal \cite{agarwal24} \emph{et al.} presented an FPTAS with a running time of $O(n^2\varepsilon^{-O(1)}\log{n})$ for the case of two square robots moving in a polygonal environment. This is the first known PTAS for two robots in such settings.

Focusing now on geodesic distance, many well-studied problems in the Euclidean setting have been explored within polygonal domains under this metric. Notable examples include geodesic convex hulls of points in polygonal domains \cite{barba_et_al:LIPIcs.SWAT.2018.8,toussaint1986relative}, the 1- and 2-center problems \cite{Pollack1989,oh2019geodesic}, and geodesic spanners~\cite{abam15}. 

Finally, Abam and Seraji \cite{abam21} constructed an $8\sqrt{3}$-spanner of size $O(n\log^3 n)$ amid axis-parallel boxes in $\mathbb{R}^3$, which is the only known geodesic spanner result in 3-dimensional space with obstacles. 

\subsection{Overview of the proof of \Cref{thm:approxrealizemetric}}

Our approach to proving \Cref{thm:approxrealizemetric} is inspired by the work of Dumitrescu and T\'oth~\cite{DumitrescuT23}, who studied a similar problem in the planar setting. They denote by \(\varrho(P)\) the ratio between the geodesic and the Euclidean diameter of a polygon \(P\), and show that over all convex polygons with \(h\) convex holes, the supremum of \(\varrho(P)\) is between \(\Omega(h^{1/3})\) and \(O(h^{1/2})\). Crucially, their construction includes building a ``wall'' of polynomially many short segments that lengthen the geodesic diameter by a large factor, see Figure~\ref{fig:dt}(i) for an illustration.

\begin{figure}
\begin{center}
\includegraphics{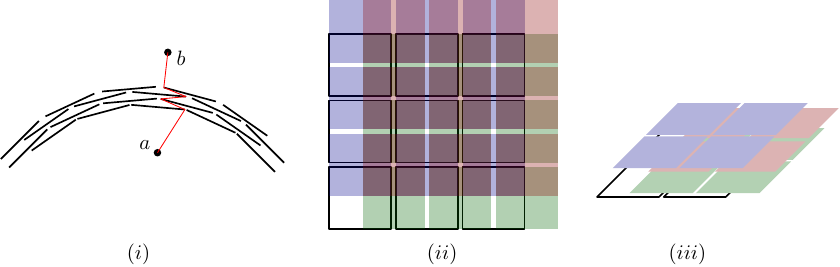}
\end{center}
\caption{(i) The construction of Dumitrescu and T\'oth~\cite{DumitrescuT23} for a wall in $\Reals^2$ consisting of convex obstacles. The shortest path from $a$ to $b$ must zigzag between the layers of segments. (ii) The view of a flat wall from $z=\infty$, with 4 shifts (iii) The flat wall construction viewed from a generic point. The planes containing the four types of squares have slightly differing $z$-coordinates, and they are sandwiched between the planes $z=0$ and $z=0.05$.}
\label{fig:dt}
\end{figure}

Our strategy to embed a metric space is as follows. We will define a custom closed surface $\cS$ embedded in $\Reals^3$, and map each point $x_i\in X$ to some point $f(x_i)$ inside. We then place a large collection of obstacles near the surface $\cS$ (within distance $1/2$) that ensures that any geodesic that would connect the inside and outside of $\cS$ must be very long, longer than the maximum distance in $(X,\dist_X)$. This means that the minimum distance between a point $f(x_i)$ and $f(x_j)$ must be realized by a geodesic that stays within $\cS$. By setting up $\cS$ correctly, we are able to give an upper and lower bound on the length of any geodesic connecting $f(x_i)$ and $f(x_j)$ so that we can ensure $\dist_X(x_i,x_j) \leq \dist_{\cT}(f(x_i),f(x_j))\leq  (1+\eps)\cdot \dist_X(x_i,x_j)$.\\

\noindent \textbf{Construction of the surface $\cS$.} Before discussing our obstacle placement, we first describe the construction of $\cS$ and the embedding $f$. (Note that the main text will define the surface $\cS$ last.) First, without loss of generality we assume that the minimum distance in $(X,\dist_X)$ is at least some large polynomial in $n:=|X|$, as we can use a scaling on the final construction to get the desired distance. The function $f$ simply assigns the points $x_i$ to equally spaced points along the $x$-axis. The surface $\cS$ will contain a cube $\cR_i$ centered at each point $f(x_i)$.

\begin{figure}
\begin{center}
\includegraphics[scale=0.8]{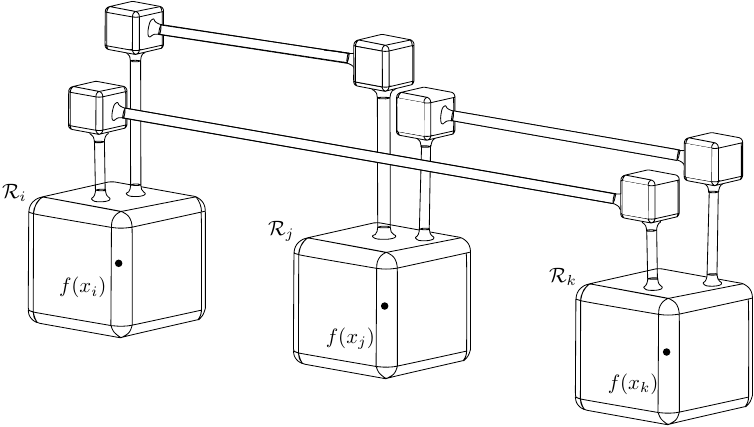}
\end{center}
\caption{The surface $\cS$ for a $3$-point metric space. The shortest geodesic from $f(x_i)$ to $f(x_j)$ is forced to go through some tubes, consisting of two vertical and a horizontal cylinder. The length of such a path can be adjusted by changing the length of the vertical cylinders.}
\label{fig:construction}
\end{figure}

For each pair $i,j$ we realize the distance of $\dist_X(x_i,x_j)$ via a sequence of three tubes (cylinders), where the first and third cylinder are vertical (parallel to the $z$-axis) and the second is parallel to the $x$-axis, see Figure~\ref{fig:construction}. The tubes connect a hole in the top of the cube $\cR_i$ to a hole on the cube $\cR_j$, and in order to realize the turns, we use two small cubes with holes that connect the first and second as well as the second and third tube in the sequence. Notice that the the length of the second horizontal tube is determined by $f(x_i)$ and $f(x_j)$, thus in order to set the length of the cube sequence, we can vary the length of the two vertical tubes. We place these tubes at different depths ($y$-coordinates) so that they remain disjoint from each other. One can see that this leads to a good approximate realization, however, this surface needs to be smoother (differentiable) in order to ease the placement of obstacles. Thus $\cS$ is defined using differentiable surface patches: we need a rounded cube and a differentiable joint that allows the connection of a flat cube face with a circular cylinder boundary.

\noindent\textbf{A flat wall example.} 
Let us consider an example of a simple placement of convex objects that increases the geodesic distance between the half-space $z<0$ and the half-space $z>1$ to significantly more than $1$. We call this construction a \emph{flat wall}. In this construction we use squares rather than triangles. Let $S_1$ be the set of axis-parallel squares of side length $0.99$ centered at the points $(a,b,0.01)$ for $a,b \in \mathbb{Z}$ in the $z=0.01$ plane. Repeat the same construction $3$ more times at heights $z=0.02,0.03,0.04$ by translating $S_1$ with the vectors $(1/2,0,0.01),(1/2,1/2,0.02),(0,1/2,0.03)$, to get the sets $S_2,S_3$, and $S_4$, respectively. See Figure~\ref{fig:dt}(ii) and (iii). Consider a geodesic connecting the planes $z=0$ and $z=0.05$: its intersection point $p$ with $z=0$ will be within $\ell_\infty$-distance $0.25$ to the center of some square in one of the four shifts, and thus it will need to have length at least $0.24$ to reach the boundary of this square before it can get to $z=0.05$. Using further copies of these squares shifted to different height ranges can increase the lower bound of $0.24$ as necessary. By making more layers that are packed more densely we can achieve any separation between the halfspaces $z<0$ and $z>1$.

While the above construction is quite intuitive, we need a more careful approach to make a similar separation near an embedded closed surface. One can see that the flat construction cannot be adapted to a surface with sharp edges, so we will need to use a smoother closed surface. Even then, we will not be able to place layers of convex objects arbitrarily close together, as convex $2$-dimensional shapes will not stay inside a curved surface. The challenge is to keep our objects in different layers disjoint, while creating a strong separation.

\noindent\textbf{Curved wall construction via placing triangle layers.} 
The surface $\cS$ that we define has the crucial property that it has offset surfaces (often called parallel surfaces) that are well-defined and differentiable for any distance offset $\delta\in [-1/2,1/2]$ from $\cS$. The offset surface at distance $\delta$ is denoted by $\cS(\delta)$. Since $\cS(\delta)$ is a closed surface, we are able to talk about the inside and outside of $\cS(\delta)$, and we can also define the region between two offsets $\cS(\delta)$ and $S(\delta')$ when $\delta<\delta'$.
In our construction, the offset $\delta$ will have a role similar to the $z$-coordinate in the flat wall construction above: we will place triangles tangent to $S(\delta)$ at some given collection of points. Using a large collection of different offsets, we can ensure the required separation between the inside of $\cS(-1/2)$ and the outside of $\cS(1/2)$.

\begin{figure}
\begin{center}
\includegraphics[scale=0.9]{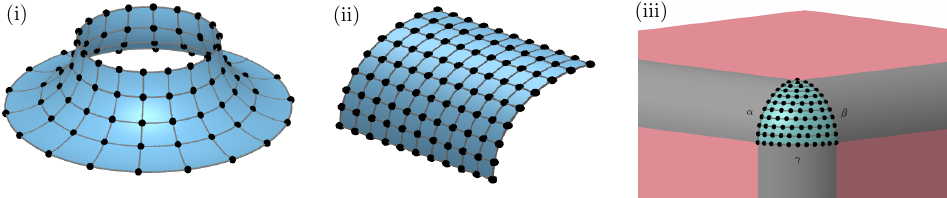}
\end{center}
\caption{ (i) A joint patch with its corresponding net. (ii) A quarter-cylinder patch with its corresponding net. (iii) A spherical triangle patch with its corresponding net.}
\label{fig:joint-cylinder}
\end{figure}


We need a handful of geometric properties from $\cS(\delta)$ to ensure that we can place our obstacles near it. The most important property is to limit their curvature\footnote{While our results could be presented with a more general differential geometric toolbox, the need to construct these $(a,b)$-nets prompted us to use specific surface patches that are easy to handle algorithmically. One could adapt our techniques to more general smooth embedded surfaces. As a trade-off our surfaces are not smooth (not even twice differentiable) at the shared boundary of our surface patches. We have also used the Euclidean distance function rather than measuring distances within these surfaces to make our construction more elementary and explicit.} in some sense. We require that at any point $p$ on $\cS(\delta)$ there exist two balls of radius at least $1/2$ whose unique intersection point with $\cS(\delta)$ and with each other is~$p$. This helps to ensure that a small $2$-dimensional convex shape of diameter $\mu\ll 1$ tangent to $\cS(\delta)$ at $p$ stays within distance $O(\mu^2)$ of $\cS(\delta)$. In particular, a small triangle will stay between the surfaces $\cS(\delta-\mu^2)$ and $\cS(\delta+\mu^2)$, and thus triangles of the same size tangent to $\cS(\delta-2\mu^2)$ and $\cS(\delta+2\mu^2)$ will remain disjoint from this triangle.

Finally, we must place the triangles in neighboring layers with ``shifts'' to ensure that a constant number of consecutive triangular layers of side-length $\mu$ triangles creates a separation of at least $\Omega(\mu)$. Here translation is not an option, so we define an \emph{$(a,b)$-net}, a weaker version of an $\eps$-net (in the metric space sense). 
For two points $u,v\in\mathbb{R}^3$ we will denote by $uv$ the line segment between $u$ and $v$, and by $|uv|$ the length of ${uv}$. An \emph{$(a,b)$-net} of a surface $\cS$ embedded in $\Reals^3$ is a point set $N\subset \cS$ where the pairwise distance of points in $N$ is at least~$a$, and for any point $q\in \cS$ there is a $p\in N$ such that $|pq|\leq b$.

We can show that the offsets of the net points provide a weaker but still useful net in the offset surfaces: a ``shared'' $(\Omega(\zeta), O(\zeta))$-net can be constructed for the surfaces $\cS(\delta)$. Here, ``shared'' means that we first find a $(\zeta,O(\zeta))$-net on $\cS$; sliding every net point along its surface normal gives, on each surface $\cS(\delta)$, an $(\Omega(\zeta), O(\zeta))$-net. These net points are then partitioned into constantly many classes, where the distances within points in each class $C$ are large enough so that equilateral triangles of side length $\mu=\Theta(\zeta)$ tangent to~$\cS(\delta)$ at each point of~$C$ will remain disjoint; i.e., these triangles would correspond to a single square set $S_1$ from our flat wall example. Using all of these partition classes in different offsets $\cS(\delta+i\zeta)$ ensures a separation of $\Omega(\zeta)$ for geodesics passing between these layers, as seen with the four shifts of squares in our flat wall example. In our construction we must set $\zeta$ small enough to be able to accommodate several iterations of the above construction between $\cS(-1/2)$ and $\cS(1/2)$, but large enough to get the desired length lower bound for any geodesic going from the inside of $\cS(-1/2)$ to the outside of $\cS(1/2)$. This concludes the generic construction of the curved wall, as well as the overview of the ideas behind~\Cref{thm:approxrealizemetric}.


\section{Doubling dimension among convex fat disjoint obstacles}\label{sec:doubling}

We denote by $B(p,r)$ the Euclidean ball of radius $r$ centered at $p\in \mathbb{R}^d$, and by $B_\mathcal{T}(p,r)$ the geodesic ball of radius $r$ centered at  $p\in \cF_\cT$.

We recall the following key result by Chew~\etal~\cite{chew2002} on convex fat objects. Based on this, we will subsequently show that the free space induced by a set of convex,fat and disjoint obstacles in $\mathbb{R}^d$ has bounded doubling dimension, under the shortest-path metric.

\begin{lemma}\label{le:boundary-path}
Let $\cT$ be a set of convex, $\alpha$-fat and disjoint obstacles in $\Reals^d$. Then, for any $u,v\in \cF_\cT$, we have that $\dist_\cT(u,v)\leq \beta \cdot |uv|$, where $\beta$ is a constant depending on $d$ and $\alpha$ defined as follows:
\[
\beta =
\begin{cases}
1 + \frac{4}{\pi\alpha}, & \text{for } d = 2, \\
1 + \frac{8d^d}{\alpha}, & \text{for } d \geq 3.
\end{cases}
\]
\end{lemma}

Since a set $\cT$ of convex, fat and disjoint obstacles distorts distances by a constant factor, it follows that the doubling dimension of $(\cF_\cT,\dist_\cT)$ will be bounded; we provide a proof for completeness.
\begin{lemma}\label{le:doubling-space}
    Let $\cT$ be a set of convex, $\alpha$-fat and disjoint obstacles in $\Reals^d$. Then the metric space $(\cF_\cT,\dist_\cT)$ has doubling dimension $O(d(d\log d + \log(1/\alpha)))$.
\end{lemma}
\begin{proof}
Let $B_\cT(p,r)$ be a geodesic ball centered at $p \in \mathcal{F}_\mathcal{T}$. Note that $B_\cT(p,r)\subseteq B(p,r)$.  We will construct a covering of $B_\cT(p, r)$ by geodesic balls of radius $r/2$. Towards that, we cover $B(p, r)$ by $c_d$ Euclidean balls of radius $r/2$, where $\log(c_d)=\Theta(d)$ is the doubling dimension of $\mathbb{R}^d$~\cite[Ch. 10]{heinonen2001metric}. Repeating this process recursively $\log(2\beta) + 1$ times, the ball $B(p, r)$ is covered by $c_d^{\log(2\beta)+1}$ Euclidean balls of radius $r / (2^{\log(2\beta)+1}) = r / (4\beta)$. Let $\mathcal{B}$ denote the resulting collection of Euclidean balls. Since $B_\cT(p,r)\subseteq B(p,r)$, this collection also covers $B_\cT(p,r)$.

We transform $\cB$ into a collection $\mathcal{B}_\cT$ of geodesic balls  in $(\mathcal{F}_\mathcal{T}, \mathrm{dist}_\mathcal{T})$ of radius $r/2$, such that the union of these geodesic balls covers $B_\cT(p, r)$. For each Euclidean ball $B(q,r/(4\beta)) \in \mathcal{B}$, we handle the following cases:

1. If $B(q, r / (4\beta)) \cap \mathcal{F}_\mathcal{T} = \emptyset$, then we ignore $B(q, r / (4\beta))$ in the cover.

2. If $B(q, r / (4\beta)) \cap \mathcal{F}_\mathcal{T} \neq \emptyset$, let $o \in B(q, r / (4\beta)) \cap \mathcal{F}_\mathcal{T}$. Include the geodesic ball $B_\cT(o, r/2)$ in $\mathcal{B}_\cT$. To verify coverage, note that if $s \in B(q, r / (4\beta)) \cap \mathcal{F}_\mathcal{T}$, then:
   \[
   |os| \leq \frac{r}{2\beta},
   \]
   By Lemma~\ref{le:boundary-path}, the geodesic distance between $o$ and $s$ satisfies:
   \[
   \mathrm{dist}_\mathcal{T}(o, s) \leq \beta \cdot |os| \leq \beta \cdot \frac{r}{2\beta} = \frac{r}{2}.
   \]
   Hence, $s \in B_\cT(o, r/2)$, and $B(q, r / (4\beta)) \cap \mathcal{F}_\mathcal{T} \subseteq B_\cT(o, r/2)$.

Each ball in $\mathcal{B}_\cT$ has radius $r/2$, and the union of these geodesic balls covers $B_\cT(p, r)$. The total number of geodesic balls in $\mathcal{B}_\cT$ is at most $c_d^{\log(2\beta)+1}$. Therefore, the doubling dimension of $(\mathcal{F}_\mathcal{T}, \mathrm{dist}_\mathcal{T})$ is given by 
\[
\log(c_d^{\log(2\beta)+1}) = (\log(2\beta)+1)\cdot\log(c_d) = O(d(d\log d + \log(1/\alpha))).\qedhere
\]
\end{proof}

\section{Realizing metric spaces with convex disjoint obstacles} \label{sec:main}

This section contains the main result of our paper. We show that by using only convex and disjoint obstacles in $\Reals^3$, it is possible to realize any metric space approximately. 

\subsection{Surface patches}\label{sec:patches}

We will consider surfaces embedded in $\Reals^3$ that are differentiable (i.e., have well-defined normals at each point) described by differentiable surface patches. More precisely, our surface $\cS$ will consist of the following types of surface patches:
\begin{itemize}
    \item Axis-parallel squares of side length at least $1$ with 0 or more disjoint circular holes of radius $2$, such that each hole is at distance at least $2$ from the square, and the hole centers have pairwise distance at least $6$. We will refer to such a patch as a \emph{square patch}.
    \item Sections of a sphere of radius $1$ cut out by axis-parallel planes through the center of the sphere. These sections correspond to octants of spheres and will be referred to as \emph{spherical triangle patches}.
    \item Cylinders of axis length at least $1$ and radius $1$, or their sections of angle $\pi/2$ around their axes, where the cylinder axis is parallel to some coordinate axis. We will refer to such a patch as a \emph{cylinder patch} and a \emph{quarter-cylinder patch}, respectively.
    \item \emph{Joints}, which are defined as a surface of revolution obtained by rotating the quarter circular arc $\{(x,y) \mid x^2+y^2=1, x\geq 0, y\geq 0\}$ around the axis $x=2$. We only allow rotations where the axis of revolution is parallel to some coordinate axis. We will refer to such a patch as a \emph{joint patch}.
\end{itemize}
Notice in particular that a joint patch can be used to connect a circular hole of a square patch to a cylinder patch, and the quarter cylinder  and spherical triangle patches can be used to "round" an axis parallel cube, i.e., to get the boundary of $C\oplus B_1$, where $C$ is an axis-parallel cube and $B_1$ is a ball of radius $1$, and $\oplus$ denotes the Minkowski sum.

We say that a surface $\cS$ is a \emph{patchwork} if it is a closed differentiable surface that is the union of some patches of the above types.

Let $\cS$ be a patchwork, and let $\bn: \cS\rightarrow \sph^2$ denote its normal bundle. For a fixed $\delta\in (-1/2,1/2)$ we define the \emph{offset surface of $\cS$ at distance $\delta$}, denoted by $\cS(\delta)$, as $\cS(\delta):=\{p+\delta\bn(p)\mid p\in \cS\}$. Because of the restriction of $|\delta|< 1/2$ we have that $\cS(\delta)$ is also a differentiable surface, and in fact it consists of patches of the following types, which we will call $\delta$-patches:
\begin{itemize}
    \item Square patches (with or without holes of radius $2$).
    \item Sections of a sphere of radius $1+\delta$ cut out by axis-parallel planes through the center of the sphere.  
    \item Cylinders of radius $1+\delta$ or their sections of angle $\pi/2$ around their axes, where the axis is parallel to some coordinate plane.   
    \item Joints, which are obtained by rotating the quarter circular arc $\{(x,y) \mid x^2+y^2=(1-\delta)^2, x\geq 0, y\geq 0\}$ around the axis $x=2$; again we only allow rotations where the axis of revolution is parallel to some coordinate axis.
\end{itemize}

 We note that when $\cS(\delta)$ is a bounded closed surface of the above type (i.e., it is finite and has no boundary), $\Reals^3\setminus \cS(\delta)$ consists of two disjoint connected components, one of which is bounded, called \emph{the inside of $\cS(\delta)$}, and the other is unbounded, called \emph{the outside of $\cS(\delta)$}. Moreover, when $\delta<\delta'$, then $\cS(\delta)$ is contained in the inside of $\cS(\delta')$, and the set of points \emph{between} $\cS(\delta)$ and $\cS(\delta')$ are those that are outside $\cS(\delta)$ and inside $\cS(\delta')$.

We call a differentiable surface $\cS$ \emph{$r$-touchable}\footnote{On a smooth surface this would correspond to bounding the principal curvature in the interval $[-1/r,1/r]$.} if for any $p\in \cS$ and any ball $B$ of radius $r$ tangent to $\cS$ at $p$ we have that $B$ intersects $\cS$ only at $p$. We say that a patchwork is a \emph{$1$-patchwork} if it is $1$-touchable.
Note that if $\cS$ is a $1$-patchwork then the offsets $\cS(\delta)$ of $\cS$ are $r$-touchable for any $r\leq 1-|\delta|$. In particular, since we will always have $|\delta|\leq 1/2$, all offsets $S(\delta)$ will be $1/2$-touchable.

\begin{lemma}\label{lem:between-offsets}
    Let $\cS(\delta)$ be an offset of a $1$-patchwork surface $\cS$. Let $H_p$ be the plane tangent to $\cS(\delta)$ at $p$ and let $\eps$ be such that $|\delta|<\frac{1}{2}-2\eps^2$. Then, for any $q\in H_p$ where $|pq|\leq \eps$ we have that $q$ lies between $\cS(\delta-2\eps^2)$ and $\cS(\delta+2\eps^2)$.
\end{lemma}
\begin{figure}
\begin{center}
\includegraphics{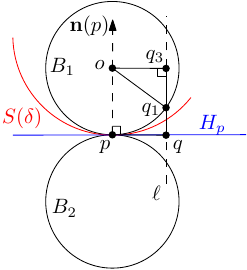}
\end{center}
\caption{The plane containing a touching segment $pq$ and the normal of $\cS(\delta)$ at $p$. The segment $pq$ will remain very close to $S(\delta)$.}
\label{fig:between-offsets}
\end{figure}
\begin{proof}
 
Note that due to the condition $|\delta|<\frac{1}{2}-2\eps^2$, the offset surfaces $\cS(\delta-2\eps^2)$ and $\cS(\delta+2\eps^2)$ are well-defined. Moreover, we have that $\frac{1}{2}-2\eps^2>0$, or equivalently that $\eps<\frac{1}{2}$.

For the following refer to Figure~\ref{fig:between-offsets}.
Consider the tangent balls $B_1,B_2$ of radius $\frac{1}{2}$ that touch $\cS(\delta)$ at $p$ and consider the line $\ell$ through $q$ parallel to $\bn(p)$. Since $|pq|=\eps<\frac{1}{2}$, the line $\ell$ intersects $B_1$. Let $q_1$ denote the point closest to $q$ in $B_1 \cap \ell$. Similarly, $\ell$ will also intersect $B_2$ and let $q_2$ denote the point closest to $q$ in $B_2 \cap \ell$. Because $\cS(\delta)$ is $1/2$-touchable, we have that the interiors of $B_1$ and $B_2$ cannot both be inside or both be outside $\cS(\delta)$, thus $q_1q_2$ intersects $\cS(\delta)$. In particular, ${q_1q}$ or ${q_2q}$ has to intersect $\cS(\delta)$; we assume without loss of generality that ${q_1q}$ intersects $\cS(\delta)$ at a point $q^*$. Since $|qq^*|<|q_1q|$, it suffices to show that $|q_1q|\leq 2\eps^2$.

Let $o$ denote the center of $B_1$ and choose $q_3$ in the extension of ${q_1q}$ towards $q_1$, such that $|q_3q|= \frac{1}{2}$. Then in the right triangle $\triangle{oq_1q_3}$ we have:
$|pq|^2+\left(\frac{1}{2}-|q_1q|\right)^2 = \left(\frac{1}{2}\right)^2$, which gives: 
$|q_1q| = \frac{1}{2} - \sqrt{\left(\frac{1}{2}\right)^2 - |pq|^2}$. (The other root of the equation would lead to $|q_1q|>1/2$, but this is not possible as $q_1q$ is shorter than the radius of $B_1$.)
Thus: 
\[
    |q_1q| =  \frac{1}{2} \cdot \left(1-\sqrt{1-(2|pq|)^2}\right)  \leq \frac{1}{2}\cdot (1-(1-(2|pq|)^2)) =  2|pq|^2  \leq 2\eps^2
\]
where in the second inequality above, we have used that $\sqrt{1-x^2} \geq 1-x^2$, for $x\in[0,1]$.
\end{proof}

\subsection{Nets, offsets, triangle layers}\label{sec:nets}

For the remainder of this section let $\cS$ be a $1$-patchwork surface, and let $|\delta|<1/2$. We start with two technical lemmas that construct nets.

\begin{restatable}{lemma}{lemEpsnet}\label{lem:epsnet}
For any $\zeta\leq 1/8$, we can compute a $(\zeta,8\zeta)$-net of size $O(\area(\cS)/\zeta^2)$ on $\cS$ in $O(\area(\cS)/\zeta^2)$ time.    
\end{restatable}

\begin{proof}
We will construct a $(\zeta,8\zeta)$-net on $\cS$, by constructing a $(\zeta,8\zeta)$-net  for each type of patch that $\cS$ consists of. These constructions will be such that the nets of two neighboring patches of $\cS$ align along their common boundary. The union of these nets will form a $(\zeta,8\zeta)$-net for $\cS$.

First, we create nets for line segments and circular arcs. Let $ab$ be a segment of length at least $1$. Then let $k=\lfloor |ab|/\zeta\rfloor$ and place $k+1$ points at equal distances  along $ab$, so let $p_i=a+\frac{i}{k}(b-a)$ for $i=0,\dots,k$. We denote by $P_{ab}$ this collection of points. Note in particular that $a,b\in P_{ab}$ and the minimum distance of points among $P_{ab}$ is at least $\zeta$ but less than $1.2 \zeta$ (as $|ab|/\zeta\geq 8$), so they form a $(\zeta,1.2\zeta)$-net on $ab$.

Consider now a circular arc $\gamma$ of radius $r\geq \zeta$ and angle $\pi/2$. Notice that for any point pair $a,b\in \gamma$ we have that the chord $ab$ corresponds to sub-arc $\gamma(a,b)$ of length between $|ab|$ and $1.2|ab|$. Now given $\zeta<1/8$, let $k=\lfloor |\gamma|/(1.2 \zeta) \rfloor$, and place $k+1$ points at equal arc lengths along $\gamma$. More precisely, we slice $\gamma$ into $k$ equal length sub-arcs, and let $P_\gamma$ denote the $k+1$ endpoints of these sub-arcs. Again, we observe that the endpoints of $\gamma$ are included in $P_\gamma$, and two consecutive points are at arc distance between $1.2 \zeta$ and $2.4 \zeta$. Therefore consecutive points along $\gamma$ have Euclidean distance at least $\zeta$ and at most $2.4\zeta$. We note moreover that for arcs of radius at least $1$ we have that (due to $\zeta\leq 1/8$) the maximum Euclidean distance of two consecutive points is less than $2\zeta$, so we get a $(\zeta,2\zeta)$ net when $r\geq 1$.

We are now ready to construct the $(\zeta,8\zeta)$-net on each patch type.
It will be convenient to denote, for a set of elements $A$ and a set of elements $B$, the set of intersections of elements from $A$ with elements from $B$ as follows: $I(A,B) = \{a\cap b: a\in A, b\in B\}$.

\begin{figure}[b]
\begin{center}
\includegraphics[scale=0.8]{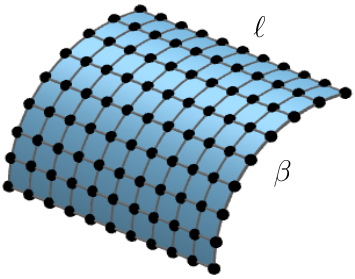}
\end{center}
\caption{ The net constructed on a quarter-cylinder patch.}
\label{fig:quarter-cylinder}
\end{figure}

\begin{itemize}

    \item \textbf{Quarter-cylinder patches.}
    
    Let $\beta$ denote one of the two boundary quarter-circular arcs of a quarter-cylinder $\cQ$. Let $\ell$ denote one of the two boundary segments of $\cQ$, as shown in Figure~\ref{fig:quarter-cylinder}. We place the point set $P_\ell$ on $\ell$ to form a $(\zeta,1.2\zeta)$-net, and the set $P_\beta$ on $\beta$ to form a $(\zeta,2\zeta)$-net. Let $B$ denote the translates of $\beta$ along $\ell$ to the points of $P_\ell$. Similarly, let $L$ denote the translates of $\ell$ along $\beta$ to the points of $P_\beta$.
    The set of points  $I(B,L)$ forms a $(\zeta,4\zeta)$-net on $\cQ$: by construction, any two arcs in $B$ and any two segments in $L$ have distance at least $\zeta$ and at most $2\zeta$. Therefore, any two points in $I(B,L)$ have distance at least $\zeta$ and from any point on $\cQ$ there is a point in $I(B,L)$ within distance $4\zeta$. The resulting net is shown in Figure~\ref{fig:quarter-cylinder}.  By construction, we have that $|I(B,L)| = |B|\cdot|L|= O(\frac{|\ell|\cdot|\beta|}{\zeta^2}) = O(\area(\cQ)/\zeta^2)$.

    \item \textbf{Cylinder patches.} 
    
    We obtain a $(\zeta,4.8\zeta)$-net for a cylinder patch by taking the union of four quarter-cylinder patches, such that their nets align along their common boundaries. By construction, the resulting set of points on the cylinder patch is a $(\zeta,4\zeta)$-net of the desired size.

    \item \textbf{Joint patches.} 
    
    Consider a joint $\cJ$ and let $\gamma,\gamma'$ denote its two circular boundaries, where $\gamma$ has radius $1$ and $\gamma'$ has radius $2$.  We place net points on $\gamma$ in the same way as for the base of a cylinder patch, which defines the point set $P_\gamma$. Let $\beta_0$ denote the quarter-circular arc that we rotated to obtain $\cJ$. For each point $p$ placed on $\gamma$, we draw on $\cJ$ the copy of $\alpha$ that goes through $p$. Let $B$ denote the resulting set of arcs.
    
    On each arc $\beta \in B$ we place a corresponding $(\zeta,2\zeta)$-net $P_\beta$. Let $A$ denote the set of circles on the patch with planes parallel to $\gamma$ and $\gamma'$ containing the net points. See Figure~\ref{fig:cube-joint} (ii).
    
    We will prove that the resulting point set $I(A,B)$ is a $(\zeta,8\zeta)$-net on $\cJ$.
    Let $p\in I(A,B)$ and let $p = \alpha\cap \beta$ for some $\alpha \in A$ and $\beta\in B$. Notice that any other net point $q$ will be either on a different circle $\alpha'\in A$ or a different circle $\beta'\in B$. Since different arcs and different circles have distance at least $\zeta$, we have that the distance of $p$ to any other point of $I(A,B)\setminus \{p\}$ is at least $\zeta$. This establishes the first property of the net. 
    
    For the second property, consider any point $p \in \cJ$ and let $\beta_p$ be the rotated copy of $\beta_0$ passing through $p$. Since any two consecutive circles in $A$ have Euclidean distance at most $2\zeta$, we know that $p$ is within distance $2\zeta$ to some point $p'\in \alpha$ where $\alpha\in A$. Since the points on $\alpha$ can be obtained from $\gamma$ by a scaling of factor at most $2$, we have that there is some point $p''\in I(A,B)$ within distance $2\cdot 2\zeta=4\zeta$ from $p'$. (We remark here that net points in $\alpha$ already form a $(\zeta,4\zeta)$-net with respect to $\alpha$.) Therefore, from $p$ we can reach a point in $I(A,B)$ within distance $6\zeta$. As a result $|pp''|\leq 6\zeta$. By construction, the size of the net is $|I(A,B)| = |A|\cdot |B| = O(\frac{|\gamma|\cdot|\beta_0|}{\zeta^2}) = O(1/\zeta^2)$. Since $\area(\cJ)=O(1)$, the size of the net is $O(\area(\cJ)/\zeta^2)$.

\begin{figure}
\begin{center}
\includegraphics[scale=0.8]{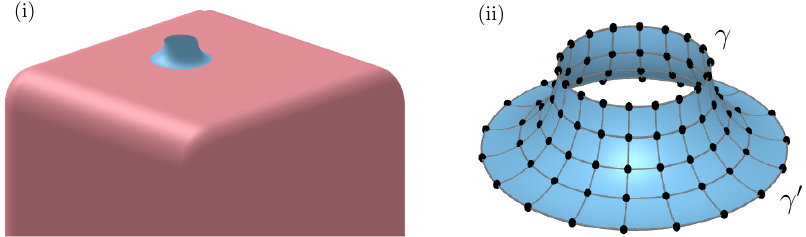}
\end{center}
\caption{(i) A joint (light blue color) on a square face of a rounded cube (light red color). (ii) The resulting net on the joint. }
\label{fig:cube-joint}
\end{figure}

\begin{figure}[t]
\begin{center}
\includegraphics{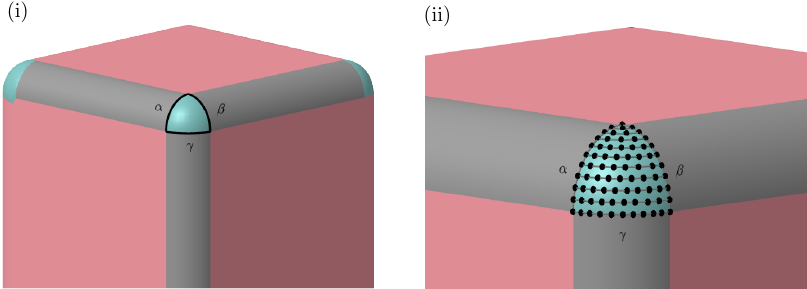}
\end{center}
\caption{(i) A spherical triangle patch (light blue color), connecting to three quarter-cylinder patches (drawn in grey). (ii) The resulting net on the spherical triangle.}
\label{fig:octant}
\end{figure}

    \item \textbf{Spherical triangle patches.}   
A spherical triangle patch $\cO$ is bounded by three quarter-circular arcs $\alpha,\beta,\gamma$ each of which connects $\cO$ to a quarter-cylinder patch; see Figure~\ref{fig:octant}(i). We place net points on $\alpha, \beta$ and $\gamma$ in the same way as for the boundary arcs of a quarter-cylinder patch; as a result, the distance of consecutive points on the same arc is between $\zeta$ and $2\zeta$. For each net point $p$ placed on $\alpha$, we draw on $\cO$ a circular arc parallel to the plane of $\gamma$ and going through $p$. Let $A$ denote the resulting set of arcs, together with $\gamma$. For an arc $\psi\in A$ let $P_\psi$ denote the $(\zeta,2.4\zeta)$-net we can make on this quarter circle. See Figure~\ref{fig:octant}(ii) for an illustration. We now show that the set of points $P_\zeta$ created this way is a $(\zeta,8\zeta)$-net on $\cO$. 

Two points of $P_\zeta$ on the same arc $\psi\in A$ have distance at least $\zeta$ by definition. For a pair of points on distinct arcs of $A$, we can lower bound their distance by the minimum distance of their arcs $\psi,\psi'$ in $A$. This distance is also realized between the points $\alpha\cap \psi, \alpha\cap \psi'$. Since the net points on $\alpha$ already have a distance of at least $\zeta$ from each other, this establishes the first property of the net. Note that consecutive arcs of $A$ have distance at most $2\zeta$ from each other, as their distances are realized along the arc $\alpha$, where the net points form a $(\zeta,2\zeta)$-net.

For the second property, let $q\in \cO$ and consider a horizontal plane through $p$. This plane intersects $\cO$ on a quarter-circular arc which is parallel to the arcs in $B$ and lies between two consecutive arcs in $A$. By the previous discussion, we know that from $q$ we can move to the nearest arc $\psi\in A$ within a distance of at most $2\zeta$. Since we know that net points on $\psi$ form a $(\zeta,2.4\zeta)$-net on $\psi$, we can then move to the nearest point in $\psi\cap P_\zeta$ within a distance of at most $2.4\zeta$. Therefore $\textrm{dist}(q,P_\zeta)\leq 8\zeta$ and thus $P_\zeta$ is a $(\zeta,4.4\zeta)$-net on $\cO$. The size of the net is $|P_\zeta| \leq \sum_{\psi\in A} |\psi|/\zeta < |A|\max_{\psi\in A} |\psi|/\zeta = O(1/\zeta^2) = O(\area(\cO)/\zeta^2)$.

    \item \textbf{Square patches.} 
    
    Let $\cK$ denote a square patch. The four edges of $\cK$ already contain net points shared with four quarter-cylinder patches. Moreover, the circular holes of radius $2$ on $\cK$ already contain net points shared with joint patches, and on each hole boundary these form a $(\zeta,4\zeta)$-net. We construct a net on $\cT$ by first placing a regular grid $G$ on $\cK$ generated by the $(\zeta,1.2\zeta)$-nets $P_{ab}$ and $P_{bc}$ on two adjacent sides $ab$ and $bc$. Next, we remove any grid points that fall inside a hole or within distance $\zeta$ from a net point on the boundary of a hole. The resulting net $P_\cK$ can be seen in Figure~\ref{fig:square-grid}. By construction, net points are at distance at least $\zeta$ from each other. For the second property of the net, first observe that any point $p\in \cK$, is within distance $1.2\cdot \sqrt{2} \zeta<2\zeta$ from a grid point $g\in G$. If $g$ was not subsequently removed, then we are done. Otherwise, there is a net point $q$ on the boundary of a hole, such that $\textrm{dist}(g,q)\leq \zeta$. Therefore, $\textrm{dist}(p,q)\leq 3\zeta$ and thus this is a $(\zeta,3\zeta)$-net on $\cK$.
    Finally, we have that $|P_\cK\cap G|=O(\area(\cK)/\zeta^2)$, and each hole has $O(1/\zeta)$ boundary net points. Due to the large pairwise distance of holes, the area of holes is less than the area of $\cK$. Since each hole has area $4\pi$, we have that there are $O(\area(\cK))$ holes, so they contribute $O(\area(\cK))/\zeta$ net points. Therefore $|P_\cK|=O(\area(\cK)/\zeta^2)$, as required. \qedhere
\end{itemize}    
\end{proof}

\begin{figure}
\begin{center}
\includegraphics[scale=0.8]{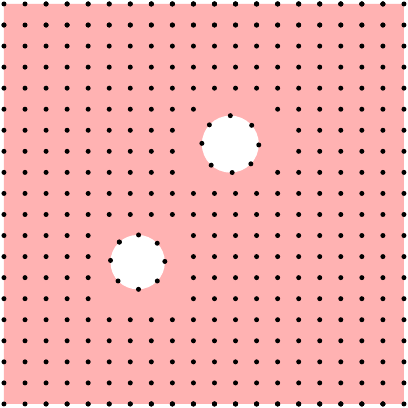}
\end{center}
\caption{A square patch with two disk holes and net points.}
\label{fig:square-grid}
\end{figure} 

\begin{restatable}{lemma}{lemEpsnetShift}\label{lem:epsnet_shift}
    Let $\zeta\leq 1/8$ and let $N\subset \cS$ be the $(\zeta,8\zeta)$-net of $\cS$ constructed above. Then, for any $\delta\in(-1/2,1/2)$, we have that $N_\delta:=\{p+\delta \bn(p)| p\in N\}$ is a $(\zeta/2,12\zeta)$-net of $\cS(\delta)$.
\end{restatable}

\begin{proof}
Let $\cP$ denote a spherical triangle patch, a quarter-cylinder patch, a cylinder patch or a square patch of $\cS$. Then the $(1,\delta)$-patch $\cP(\delta)$ of $\cS(\delta)$ is a scaling of $\cP$ by a factor of ${1+\delta}$ (or $1$ in case of the square patch). Therefore, the Euclidean distance of any two points on $\cP$ is scaled by a factor $1/2\leq 1+\delta \leq 3/2$. As a result, the set $\cP(\delta)\cap N_\delta$ is a $(\zeta/2,12\zeta)$-net of~$\cP(\delta)$. 

Let $\cJ$ be a joint patch of $\cS$ and let $\gamma,\gamma'$ denote its top and bottom circle respectively. Observe that in the joint $(1,\delta)$-patch $\cJ(\delta)$ of $\cS(\delta)$, $\gamma'$ remains a circle of radius $2$, while $\gamma$ becomes a circle of radius $1+\delta$.  Moreover, the arc $\alpha_\delta$ that is being rotated is a quarter-circular arc of a circle with radius $1-\delta$. Therefore, distances along the top circle of $\cJ(\delta)$ and along the rotated arc, are scaled by a factor $1/2\leq 1 \pm \delta \leq 3/2$. By following the proof for the joint patch in Lemma~\ref{lem:epsnet} we see that $N_\delta\cap \cJ(\delta)$ is a $(\zeta/2,12\zeta)$-net of $\cJ(\delta)$. 
\end{proof}

\begin{restatable}{lemma}{lemEpsnetNeighbors}\label{lem:epsnet_neighbors}
    Let $N_\delta$ be a $(\zeta/2,12\zeta)$-net of $\cS(\delta)$ for  each $|\delta|<1/2$ constructed above. Then for any fixed~$\delta$, any $1\leq t<1/(4\zeta)$ and any $p\in N_\delta$ we have that the ball $B(p,t \zeta)$ contains at most $O(t^2)$ points of $N_\delta$. Moreover, for any $t$ we can partition $N_\delta$ into $O(t^2)$ point sets $N_{\delta,i}$ such that the pairwise distance of points within $N_{\delta,i}$ is at least $t\zeta$.
\end{restatable}

\begin{proof}
    First note that $t\cdot \zeta < 1/4$. Moreover, any cylinder patch of $\cS(\delta)$, has distance at least $1/2$ from any square patch of $\cS(\delta)$. Therefore, $B(p,t\cdot \zeta)$ cannot intersect both a square and a cylinder patch at the same time. Recall that any two holes on a square patch have pairwise distance at least $6$ and are also at distance at least $2$ from the boundary of the square patch. Therefore, $B(p,t\cdot \zeta)$ cannot intersect both a quarter-cylinder patch and a joint or a cylinder patch at the same time. The highest number of distinct surface patches $B(p,t\cdot \zeta)$ can intersect is $7$, which may happen when it intersects a spherical triangle patch and the three quarter-cylinder patches around it and the three square patches around it. Since it can intersect at most a constant number of distinct patches, it suffices to show that it can contain at most $O(t^2)$ points of $N_\delta$ from each different type of surface patch.

   To that end, let $\cP$ denote a patch which is not a square patch and not a spherical triangle patch. For the remaining types of patches, the constructed net is taken as a point set of the form $I(A,B)$, where elements in $A$ and $B$ are quarter-circular arcs, line segments or circles. In all cases, we have argued already that any two elements in $A$ have distance at least $\zeta/2$ (since now $N_\delta$ is a $(\zeta/2,12\zeta)$-net) and any two elements in $B$ have also distance at least $\zeta/2$. Therefore, from any point $p\in \cP$ and within distance $t\cdot\zeta$, we can reach at most $2t$ elements of $A$ and $2t$ elements of $B$. Therefore we can reach at most $4t^2$ net points of $\cP$. On a spherical triangle patch, the arcs of $A$ again have pairwise distance at least $\zeta/2$, so we can reach at most $O(t)$ arcs, and on each arc we can reach at most $O(t)$ points as on each quarter-circular arc the distance between consecutive points is at least $\zeta/2$.

   For a square patch, the process of construction was slightly different due to the existence of the holes. Using the same argumentation as above, we can deduce that $B(p,t\cdot \zeta)$ can contain at most $4t^2$ net points from the initially placed grid. Moreover, $B(p,t\cdot\zeta)$ can intersect at most one hole. Note that a single hole contains $O(1/\zeta)$ net points and since $t<1/(2\zeta)$, it contains at most $O(t)$ net points. Therefore, $B(p,t\cdot\zeta)$ can contain at most $5t^2$ net points of a square patch. 

   Since we argued that $B(p,t\cdot \zeta)$ can intersect at most $7$ distinct surface patches, we have that it can contain at most $7\cdot5t^2 = 35t^2$ points of $N_\delta$.  

    For the second part of the lemma create $b=\lfloor{35t^2+1}\rfloor$ partition classes $N_{\delta,1},N_{\delta,2},...,N_{\delta,b}$ and greedily place points in these classes: we place each point $p\in N_{\delta}$ in the available class of smallest index (that is, a class of minimum index not already containing any point $q$ such that $|pq|<t\cdot \zeta$). Note that in this way, there will always be a class which is available, since for any point there are at most $35t^2$ points that cannot be in the same class with it.
\end{proof}

\begin{restatable}{lemma}{lemSeparation}\label{lem:separation}
There exists a constant $\nu$ such that for any $\zeta<1/(2\nu)$ and $-1/2<\delta<1/2-\nu \zeta^2$ there is collection $\cT$ of $O(\area(\cS)/\zeta^2)$ congruent equilateral triangles of side length $\Theta(\zeta)$ that   satisfy the following properties:
\begin{enumerate}[(i)]
    \item each triangle in $\cT$ is located between $\cS(\delta)$ and $\cS(\delta+\nu \zeta^2)$

    \item the triangles in $\cT$ are pairwise disjoint, and
    
    \item the geodesic distance from any point of $\cS(\delta)$ to any point of $\cS(\delta+\nu \zeta^2)$ is at least $4 \zeta$.
\end{enumerate}
\end{restatable}

\begin{proof} 
By applying \Cref{lem:epsnet} and \Cref{lem:epsnet_shift} for we obtain a $(\zeta/2,12\zeta)$-net $N_\delta$ of $\cS(\delta)$ of size $O(\area(\cS)/\zeta^2)$. We set $t=48(1+\sqrt{3})$, and choose $\nu$ large enough such that $\zeta<1/(4t)$. Then, by Lemma~\ref{lem:epsnet_neighbors}, we can partition $N_\delta$ into $b = \lfloor{35t^2+1}\rfloor$ classes $(N_{\delta,i})_{i=1}^{b}$, where points within each class $N_{\delta,i}$ have pairwise distance at least $t\cdot \zeta = 48(1+\sqrt{3})z$. 
We will now place each class $N_{\delta,i}$ on a different offset $\cS(\delta_i)$, defined as follows. We start by setting $\nu = 4b+1$ and define
$\delta_i := \delta + 4i \cdot \zeta^2$, for $i=1,2,..,b$. From now on, we consider the points of each class $N_{\delta,i}$ on the surface $\cS(\delta_i)$ defined as above.

The above definition ensures $\delta_i \leq \delta+4b \zeta^2<\delta+\nu \zeta^2<1/2 $. Thus, the offsets $S(\delta_i)$ are well-defined. Furthermore, we also have the following: 
\begin{equation}\label{eq:property}
    \text{For $1 \leq i \leq b$, we have that $\delta_i - 2\zeta^2 < \delta_i < \delta_i + 2\zeta^2 < \delta_{i+1}$.}
\end{equation}

For $1\leq i \leq b$ and $p\in  N_{\delta,i}$, let $T_p$ denote the equilateral triangle of side length $48\sqrt{3}\zeta$ centered at $p$, which lies on the tangent plane $H_p$ at $p$. Let $\cT_i = \{ T_p: p\in \cN_{\delta,i}\}$  and $\cT = \bigcup_{i=1}^b \cT_i$. We now argue that $\cT$ satisfies the three properties of the Lemma:

\begin{enumerate}[(i)]
\item 
This holds trivially since $\delta+4\zeta^2\leq \delta_i\leq \delta+4b\zeta^2<\delta+\nu \zeta^2$.
\item Lemma~\ref{lem:between-offsets} together with \eqref{eq:property}, ensure that any two triangles in $\cT$ that belong to different offsets are disjoint. Next, we argue that any two triangles $T_{p},T_{q} \in \cT_i$ are disjoint. From Lemma~\ref{lem:epsnet_neighbors} we know that $|pq| \geq 48(1+\sqrt{3}) \zeta >96 \zeta$. Since each of our triangles has circumradius $48\zeta$, it follows that $T_p,T_q$ cannot intersect.

\item Refer to \Cref{fig:separation} for an illustration. Let $p\in \cS(\delta)$ and $q\in \cS(\delta+\nu \zeta^2)$ and consider a shortest geodesic path $\pi(p,q)$. We want to show that $\dist_\cT(p,q)\geq c\cdot \zeta$ for some constant $c$. Let $u$ be a point in $\pi(p,q)\cap \cS(\delta_1)$ and consider the segment $s$ between the points $u-\bn(u)$ and $u+\bn(u)$. We consider the following two cases:

\begin{itemize}
    \item \emph{Case I: There exists $v\in \pi(p,q)$, such that $\dist_{\Reals^3}(v,s) \geq 4\zeta$.} \\
    Then $\pi$ has the desired length.
    \item \emph{Case II: $\pi(p,q)$ is contained in the cylinder with axis $s$ and radius $4\zeta$} \\
    Let $\cC({ab},x)$ denote the cylinder of axis ${ab}$ and radius $x$.
    We will show that there exists $T\in \cT$ which cuts the cylinder $\cC(s,4\zeta)$ in such a way that the top and bottom disk of the cylinder are completely contained in different sides of the cut. Since $\pi(p,q)$ connects these disks and it is disjoint from all triangles $T\in \cT$, this will give a contradiction. 
    
    By the properties of $N_\delta$, we know that there exists a surface $\cS(\delta_i)$, such that $s\cap \cS(\delta_i)$ is within distance $12\zeta$ from a point $w\in N_{\delta,i}$. We will argue that $T_w$ is the required triangle. Towards that, let $s_i = s\cap \cS(\delta_i)$ and recall that $T_w$ has inradius $24\zeta$ and denote by $D$ the incircle of $T_w$. Observe that $\cC({s_{i-1}s_{i+1}},4\zeta) \subset B(w,24\zeta)$. By Lemma~\ref{lem:between-offsets}, we know that $T_w$ is between $\cS(\delta_{i-1})$ and $\cS(\delta_{i+1})$. Since $D$ splits $B(w,24\zeta)$ in two pieces where $s_{i-1}$ and $s_{i+1}$ are in different parts, we get that $D$ has to cut $\cC({s_{i-1}s_{i+1}},4\zeta)$. We conclude that indeed $T_w$ cuts $\cC(s,4\zeta)$ in the described manner, which is a contradiction. \qedhere
\end{itemize}
\end{enumerate}
\end{proof}
    \begin{figure}\label{fig:separation}
\begin{center}
\includegraphics{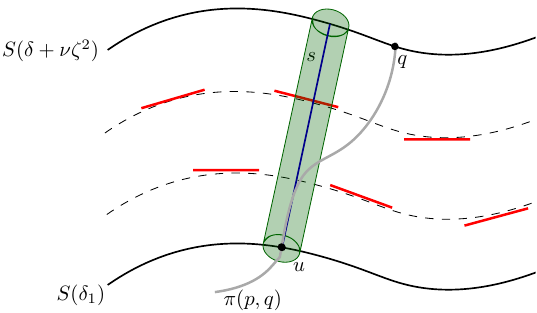}
\end{center}
\caption{Illustration for the proof of \Cref{lem:separation} (side view). The geodesic $\pi(p,q)$ is shown in grey. Note that the cylinder around $s$ is cut by one of the red triangles, therefore $\pi(p,q)$ needs to exit the cylinder.}
\label{fig:separation}
\end{figure}

By iterating the construction of Lemma~\ref{lem:separation} we get our curved wall separation theorem.

\begin{restatable}[Polynomial separation]{theorem}{thmPolySep}\label{thm:polynomial-sep}
    Let $\cS$ be a $1$-patchwork surface. Then for any $\sigma> 1$ there is a collection of $O(\sigma^4\area(\cS))$ pairwise disjoint congruent equilateral triangle obstacles between $\cS(-1/2)$ and $\cS(1/2)$, such that the geodesic distance from any point of $\cS(-1/2)$ to any point of $\cS(1/2)$ is at least $\sigma$.
\end{restatable}

\begin{proof}
We apply \Cref{lem:separation} for $\zeta = \frac{1}{2\nu\sigma_0}$ and for $\sigma_0^2$ different values of $\delta$, namely for $\delta_i = -1/2+2i\cdot\nu \zeta^2$, $i = 1,2,...,\sigma_0^2$, where $\sigma_0>1$ and its exact value will be defined later. This is possible since $\zeta<1/(2\nu)$ and $-1/2<\delta_i<1/2-\nu \zeta^2$ for the chosen values of $\zeta$ and~$\delta_i$. 

From Lemma~\ref{lem:separation}, we have that there is a collection $\cT_i$ of $O(\sigma_0^2\area(\cS))$ pairwise disjoint and congruent equilateral triangles that satisfy the three properties of the Lemma. Namely, each triangle in $\cT_i$ is between $\cS(\delta_i)$ and $\cS(\delta_i+\nu \zeta^2)$ and since $\delta_i+\nu \zeta^2< \delta_{i+1}$ (where $\delta_{\sigma_0^2+1}:= 1/2-\nu \zeta^2$), we have that the geodesic distance from any point on $\cS(\delta_i)$ to any point on $\cS(\delta_{i+1})$ is at least $\frac{4}{2\nu\sigma_0}$. A geodesic path from a point $p\in\cS(-1/2)$ to a point $q\in\cS(1/2)$, will intersect each $\cS(\delta_i)$ at least once and therefore $\dist_\cT(p,q)\geq \sigma_0^2\cdot \frac{2}{\nu\sigma_0} = \frac{2}{\nu}\cdot \sigma_0$. The collection of triangles $\cT = \bigcup_{i=1}^{\sigma_0^2} \cT_i$ has size $O(\sigma_0^2\cdot \sigma_0^2\area(\cS))=O(\sigma_0^4 \area(\cS))$ and clearly consists of disjoint, congruent equilateral triangles. Substituting $\sigma_0=\frac{\nu}{2}\sigma$ yields the desired result.
\end{proof}

\subsection{The construction and the proof of \Cref{thm:approxrealizemetric}}\label{sec:construct}

Recall that $(X,\dist_X)$ is the metric space whose distances we want to realize, with $|X|=n$ and spread $\Phi=\max_{a,b,c,d\in X}(\frac{\dist_X(a,b)}{\dist_X(c,d)})$, and $\eps>0$ denotes an accuracy parameter. We can assume that the minimum distance of any two points in $X$ is $41n^3/\eps$, as in general we can scale the whole construction so that the minimum geodesic distance exactly matches the minimum distance in $X$. We construct a $1$-patchwork surface $\cS$ together with an injection $f:X\rightarrow \mathbb{R}^3$ as follows.

For each point $x_i\in X$, $1\leq i \leq n$, let $f(x_i) = (40n^2\cdot i,0,-10n^2)$. Let $\cC_i$ denote the cube centered at $f(x_i)$ of side-length $20n^2-2$. We define the rounded cube $\cR_i = \cC_i \bigoplus B_1$, where $B_1$ denotes the ball of radius $1$. Note that each rounded cube $\cR_i$ has \emph{size} $20n^2$, where we define the size of a rounded cube to be the side-length of its bounding box. Then the top square patch of each $\cR_i$ lies on the plane $z=0$, has side length $20n^2-2$ and is centered at the point $(40n^2\cdot i,0,0)$. For each pair $(i,j)$ where $i<j$, we model the distance between $x_i$ and $x_j$ as follows. We start by placing two holes $H_{i,j},H_{j,i}$ on the top square patch of $\cR_i$ and $\cR_j$ respectively, such that: \begin{itemize}
    \item  $H_{i,j}$ has radius $2$ and center $(40n^2\cdot i,-10n^2+10(n-1)\cdot i+10j-10,0)$, and
    \item $H_{j,i}$ has radius $2$ and center $(40n^2\cdot j,-10n^2+10(n-1)\cdot i+10j-10,0)$.
\end{itemize}  
   Observe that the centers of any two holes $H_{i,j},H_{j,i}$ defined this way have the same $y$-coordinates. Moreover, the $y$-coordinates of the centers of all holes are in the range $[-10n^2+10n,-10n]$. Since the $y$-coordinates of the top square patches are in the range $[-10n^2+1,10n^2-1]$, this ensures that: 

\begin{enumerate}
    \item holes are well-defined,
    \item each hole has distance at least $2$ from the boundary of its square patch, and
    \item for each $1\leq i,j,k \leq n$, where $i,j,k$ are pairwise different, the centers of $H_{i,j}$ and $H_{i,k}$ have distance at least $10$. As a result, $H_{i,j}$ and $H_{i,k}$ have distance at least $6$, as required in the definition of a square patch.
\end{enumerate}

It will be useful for the rest of the description, to define a \emph{cylinder-joint patch of length $\ell$} as a cylinder patch of length $\ell$, with two joint patches attached to its bases. Note that a cylinder-joint patch can be used to connect two holes that are opposite to each other. 

Next, for each $1\leq i<j \leq n$, we connect $H_{i,j}$ and $H_{j,i}$ as follows. We attach to both $H_{i,j}$ and $H_{j,i}$ a cylinder-joint patch of length $\ell_\textrm{vert}(i,j)$, which we denote by $\cK_{i,j}$ and $\cK_{j,i}$ respectively. The length $\ell_\textrm{vert}(i,j)$\ will depend on $\eps$ and its exact value will be set in the proof of \Cref{thm:approxrealizemetric}.  We then attach $\cK_{i,j}$ to a rounded cube $\cC_{i,j}$ of size $8$, through a circular hole of radius $2$ centered at the center of the bottom square patch of $\cC_{i,j}$. Note that this hole has distance at least $2$ from the boundary of the square patch. Similarly, we attach $\cK_{j,i}$ to a rounded cube $\cC_{j,i}$ of size $8$, through a circular hole of radius $2$ centered at the center of the bottom square patch of the cube. We finally connect $\cC_{i,j}$ and $\cC_{j,i}$ through a horizontal cylinder-joint patch of length $\ell_\textrm{hor}(i,j) = 40n^2\cdot(j-i)$, denoted by $\cJ_{i,j}$. This concludes the  construction of $\cS$. Refer to Figure~\ref{fig:construction_} for an illustration of the connection between $f(x_i),f(x_j)$. After constructing $\cS$ we apply Theorem \ref{thm:polynomial-sep} for $\sigma := \max_{x_i,x_j}\dist_X(x_i,x_j)=O(\Phi\cdot n^3/\eps)$, to achieve a separation of $\max_{x_i,x_j}\dist_X(x_i,x_j)$ between $\cS(-1/2)$ and $\cS(1/2)$. We let $\cT$ denote the resulting set of triangle obstacles. Using the collection $\cT$, we proceed to prove \Cref{thm:approxrealizemetric}.

\begin{figure}
\begin{center}
\includegraphics[scale=0.8]{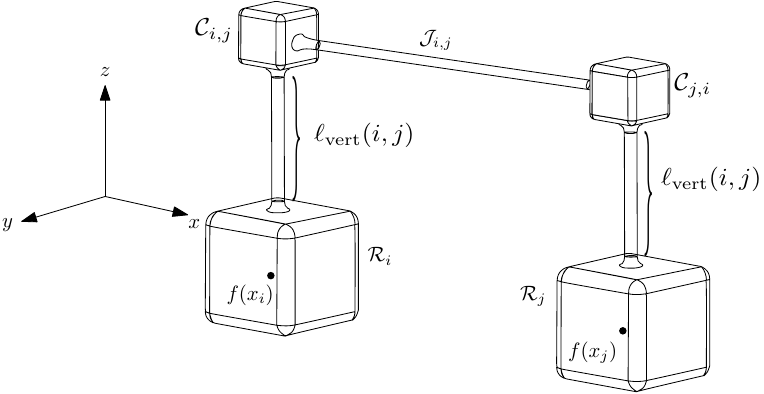}
\end{center}
\caption{A tube consisting of three cylinders and a turn that connects the rounded cubes $\cR_i$ and $\cR_j$. We will set $\ell_{\mathrm{vert}}(i,j)$ to ensure that the length of the tube corresponds to the distance $\dist_X(x_i,x_j)$.}
\label{fig:construction_}
\end{figure} 

\begin{proof}
We want to give an upper bound and a lower bound on the length of the shortest path from $f(x_i)$ to $f(x_j)$ for all $1\leq i<j \leq n$.

 We start by observing that a geodesic path shorter than $\max_{x_i,x_j}\dist_X(x_i,x_j)$ connecting $f(x_i)$ and $f(x_j)$ needs to stay inside $\cS(1/2)$, as leaving $\cS(1/2)$ requires length more than $\max_{x_i,x_j}\dist_X(x_i,x_j)$ by \Cref{thm:polynomial-sep}. A geodesic staying in $\cS(1/2)$ must do exactly one of the following:

 \begin{itemize}
     \item It will not visit any rounded cube $\cR_k$, where $k$ is different than $i,j$. In that case, we refer to this as a $1$-hop path.
     \item It will go through at least one rounded cube different than $\cR_i,\cR_j$. In that case, we refer to this path as a multi-hop path.
 \end{itemize}

 A $1$-hop shortest path has to consist of the following paths (note that here we ignore segments within $\cR_i,\cR_j$ in our consideration):

\begin{enumerate}
    \item A path $\pi_1$ connecting $H_{i,j}$ to $\cC_{i,j}$. Then $|\pi_1| \geq \ell_{\textrm{vert}}(i,j)+2$.
    \item A path $\pi_2$ between the two holes of $\cC_{i,j}$. The smallest segment between these two holes has length $2\sqrt2$.
    \item A path $\pi_3$ connecting $\cC_{i,j}$ to $\cC_{j,i}$. Then $|\pi_3|\geq \ell_{\textrm{hor}}(i,j)+2$.
    \item A path from the "left" hole of $\cC_{j,i}$ to $H_{j,i}$. By symmetry, this will have length at least $|\pi_1|+|\pi_2|$.
\end{enumerate}

Thus, a $1$-hop path has length at least $2\ell_{\textrm{vert}}(i,j)+\ell_{\textrm{hor}}(i,j)+6+4\sqrt{2}$. We set $\ell_{\textrm{vert}}(i,j)$ such that this matches exactly the distance of $x_i,x_j$ in the metric space:

\[
\dist_X(x_i,x_j) = 2\ell_{\textrm{vert}}(i,j)+\ell_{\textrm{hor}}(i,j)+6+4\sqrt{2}
\]
which leads to 
\[
\ell_{\textrm{vert}}(i,j)= \frac{\dist_X(x_i,x_j)-\ell_{\textrm{hor}}(i,j)}{2}-3-2\sqrt{2}
\]

We do this for all pairs $(i,j)$, where $1\leq i<j\leq n$. Since $\dist_X(x_i,x_j)\geq 41n^3/\eps$ and $\ell_{\textrm{hor}}(i,j) <40n^3$, we have that $\ell_{\textrm{vert}(i,j)}>1$, and thus this is a valid choice.

To lower bound $\dist_\cT(f(x_i),f(x_j))$, again we note that:
\begin{itemize}
    \item If it is a $1$-hop path, then we obtain the lower bound
\begin{equation}\label{eq:lower-bound}
\dist_\cT(f(x_i),f(x_j)) \geq 2\ell_{\textrm{vert}}(i,j) + \ell_{\textrm{hor}}(i,j) + 6 + 4\sqrt{2} = \dist_X(x_i,x_j)
\end{equation}
    \item If it is a multihop path, let $\cR_{i_1},\cR_{i_2},...,\cR_{i_\lambda}$ denote the sequence of large rounded cubes the path goes through, where $i=i_1$ and $j=i_\lambda$. Due to our choices of $\ell_{\textrm{vert}}(i,i_1), \ell_{\textrm{vert}}(i_1,i_2)$, \dots , $\ell_{\textrm{vert}}(i_{\lambda},j)$ this path will have length at least: 
    \[
    \sum_{s=1}^{\lambda-1}\dist_X(x_{i_s},x_{i_{s+1}}) \geq \dist_X(x_i,x_j),
    \]
    due to the triangle inequality in the metric space $(X,\dist_X)$. 
    Thus, Equation~\ref{eq:lower-bound} gives a lower bound also in this case.
\end{itemize}

To upper bound $\dist_\cT(f(x_i),f(x_j)))$, we construct a path  consisting of the following segments (note that here we also need to take into account segments within $\cR_i,\cR_j)$: 
\begin{enumerate}
    \item A segment between $f(x_i)$ and the center of $H_{i,j}$, which has length 
    \[\ell(i,j) = \sqrt{(10n^2)^2+(-10n^2+10(n-1)\cdot i+10j-10)^2}.\]
    \item A segment between the center of $H_{i,j}$ and the center of $\cC_{i,j}$, which has length $\ell_\textrm{vert}(i,j)+6$.
    \item A segment between the center of $\cC_{i,j}$ and the center of $\cC_{j,i}$, which has length $\ell_{\textrm{hor}}(i,j)+10$. 
    \item A segment between the center of $\cC_{j,i}$ and the center of $H_{j,i}$, which has length $\ell_\textrm{vert}(i,j)+6$.
    \item A segment between the center of $H_{j,i}$ and $f(x_j)$, which has length
    \[\ell(i,j) = \sqrt{(10n^2)^2+(-10n^2+10(n-1)\cdot i+10j-10)^2}.\]
\end{enumerate}

This path gives the following upper bound:
\[
\dist_\cT(f(x_i),f(x_j)) \leq 2\ell(i,j) + 2\ell_\textrm{vert}(i,j) + \ell_{\textrm{hor}}(i,j) + 22 
\]
 
Due to our choice of $\ell_\textrm{vert}(i,j)$, this can be rewritten as
\[
\dist_\cT(f(x_i),f(x_j)) \leq 2\ell(i,j)+\dist_X(x_i,x_j)+16-4\sqrt{2}
\]

It remains to show that $\dist_\cT(f(x_i),f(x_j)) \leq (1+\eps)\dist_X(x_i,x_j)$. It suffices to show that 
\[
2\ell(i,j)+\dist_X(x_i,x_j)+16-4\sqrt{2} \leq (1+\eps)\dist_X(x_i,x_j)
\]

or that 
\[
\eps\cdot\dist_X(x_i,x_j)\geq 2\ell(i,j)+16-4\sqrt{2}
\]

This holds since $\ell(i,j) = O(n^2)$ and $\eps\cdot\dist_X(x_i,x_j)\geq 41n^3$. 

It remains to bound the number of triangular obstacles in $\cT$. From Theorem~\ref{thm:polynomial-sep}, we need a bound on $\area(\cS)$, which we obtain by putting together the following:

\begin{itemize}
    \item The total surface area of the cylinders is ${n \choose 2} \cdot O(\mathrm{max}_{x_i,x_j\in X}(\dist_X(x_i,x_j))) = O(\frac{n^5\cdot \Phi}{\eps})$.
    \item  The total surface area of the cubes $\{\cR_i\}_{i=1}^n$ is $O(n\cdot n^4) = O(n^5)$.
\end{itemize}

The remaining parts (joints and small rounded cubes) are insignificant compared to the the are of the cylinders, thus $\area(\cS)= O(\frac{\Phi n^5}{\eps})$. Thus, from Theorem~\ref{thm:polynomial-sep} we have that $|\cT| = O(\area(\cS) \cdot \sigma^4) = O(\Phi n^5/\eps \cdot (\Phi n^3/\eps)^4)=O(n^{17}\Phi^5/\eps^{5})$.

Our construction can be followed by an algorithm whose running time is polynomial in the number of triangles placed.
\end{proof}

We note the following observation from within the proof of \Cref{thm:approxrealizemetric} for later use.

\begin{observation}\label{obs:approximate-paths}
    There is a path connecting $f(x_i)$ and $f(x_j)$ in the inside of $\cS(-1/2)$ that has length at most $(1+\eps)\cdot \dist_X(x_i,x_j)$.  
\end{observation}

\begin{remark}[Exact realizations]\label{rem:exact}
    It is particularly challenging to realize a metric \emph{exactly} with convex obstacles. While our construction can be modified to ensure the exact length of the shortest path through the corresponding tube has the same exact length as the original distance in $(X,\dist_X)$, this is not sufficient.
    
    The problem can be seen by considering the metric induced by a star on $4$ vertices. If $a,b,c$ are the leaves and $x$ is the center of the star, then the the shortest $f(a) \rightarrow f(b)$ path will be realized by going through the $f(a) \rightarrow  f(x)$ tube and the $f(x) \rightarrow f(b)$ tube, foregoing the small detour to $f(x)$ within its rounded cube. This path will be strictly shorter than $\dist_\cT(f(a),f(x))+\dist_\cT(f(x),f(b))$. While a realization exists for star graph metrics using relatively open triangles, this does not generalize. Can we for example realize $(X,\dist_X)$ exactly with disjoint convex obstacles assuming the image of $\dist_X$ is $\{1,2\}$?
\end{remark}

\section{Realizing metrics with fat obstacles}

We tweak our construction to prove the following theorem and answer \Cref{q:metricembed} for convex fat non-disjoint obstacles and for fat, disjoint but non-convex obstacles.

\begin{restatable}{theorem}{thmFat}\label{thm:fatobstacles}
    Let $\eps\in (0,1)$ and let $(X,\dist_X)$ be a metric space of size $|X|=n$ and spread $\Phi:=\max_{a,b,c,d\in X}(\frac{\dist_X(a,b)}{\dist_X(c,d)})$. Then there exists a collection $\cT$ of obstacles for each of the following restrictions:
    \begin{enumerate}[(i)]
        \item $\cT$ consists of $O(n^{17}\Phi^5/\eps^{5})$  congruent regular simplices (that are not necessarily disjoint).
        \item $\cT$ consists of  $O(n^5\Phi/\eps)$ disjoint, similarly-sized $\alpha$-fat objects that are homeomorphic to a ball. The objects have total complexity $O(n^{17}\Phi^5/\eps^{5})$, where $\alpha>0$ is a constant.
    \end{enumerate}
    In both cases, there is an injection $f:X\rightarrow \Reals^3$ such that
    \[\dist_X(a,b) \leq \dist_{\cT}(f(a),f(b))\leq  (1+\eps)\cdot \dist_X(a,b).\]
    for all $a,b\in X$.
    The set of obstacles can be constructed in $\poly(n,\Phi,1/\eps)$ time.
\end{restatable}

\begin{proof}
(i)
    We modify the set of triangular obstacles in our basic construction as follows. Recall from the proof of Lemma~\ref{lem:separation} that each triangle $T\in \cT$ is tangent at the center of $T$ to some $\cS(\delta)$ at a point $p\in\cS(\delta)$, for some $-1/2<\delta<1/2$. We turn $T$ into a regular tetrahedron $T'$ where one face of $T'$ is $T$, and the vertex of $T'$ outside $T$ is placed along the normal ray from $p$, that is, the tetrahedron will point to the outside of $\cS(\delta)$.
    The obstacles in the new collection obtained this way, denoted by $\cT'$, are convex and $1/3$-fat.

    We can now prove Theorem~\ref{thm:approxrealizemetric} for $\cT'$ instead of $\cT$. For the lower bound, it suffices to observe that $\cF_{\cT'}\subset \cF_\cT$. This implies that the set of geodesics in $\cF_{\cT'}$ is a subset of the geodesics in $\cF_{\cT}$, and thus for any two points $x_i,x_j\in X$ we have $\dist_{\cT'}(f(x_i),f(x_j))\geq \dist_\cT(f(x_i),f(x_j))\geq \dist_X(x_i,x_j)$. On the other hand, the inside of $\cS(-1/2)$ is in $\cF_{\cT'}$, due to the way we constructed the tetrahedra (pointing ``out'' of the offset surface they are tangent to). \Cref{obs:approximate-paths} now implies that $\dist_{\cT'}(f(x_i),f(x_j))\leq (1+\eps)\dist_X(x_i,x_j)$.
    
\medskip
(ii) Let $\cT$ be the triangle collection constructed by \Cref{sec:construct}, and let $\cU$ be the union of all offsets of $\cS$, that is, $\cU:=\bigcup_{\delta\in [-1/2,1/2]} \cS(\delta)$. Refer to \Cref{fig:th13-1} for an illustration. Let $N$ be a $(1/8,1)$-net of $\cS$ obtained from \Cref{lem:epsnet_shift} with $\zeta=1/8$. Consequently, $|N|=\area(\cS)=O(n^5\Phi/\eps)$. Compute the \emph{Voronoi diagram} of $N$ in $\Reals^3$, that is, a partition of $\Reals^3$ into cells according to the closest point of $N$. The Voronoi cells can be computed in $\poly(|N|)$ time~\cite[Ch. 27]{MR1730156}. For the cell $C_p$ of $p\in N$ let $U_p:=C_p\cap \cU$; note that each $U_p$ can be represented explicitly in $\poly(|N|)=\poly(|\cT|)$ space.
    Observe that $U_p$ has diameter at most $\diam(U_p)\leq 3/2$: indeed, any point in $\cU$ is within distance at most $1/2$ from $\cS$ (along the normal direction or its inverse), and any point of $\cS$ is within distance at most $1$ from some point of $N$ by the net property of~$N$.

\begin{figure}
\begin{center}
\includegraphics[scale=0.95]{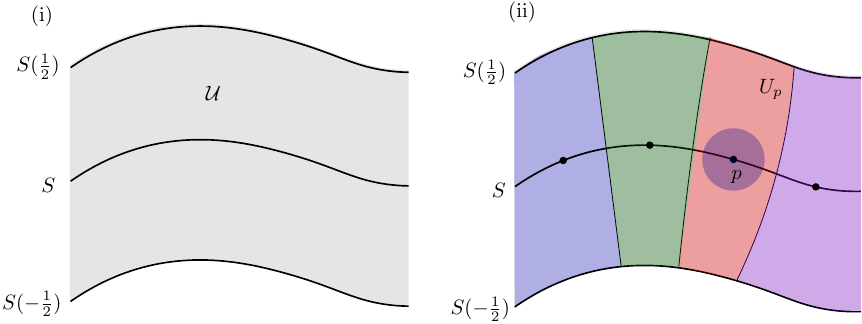}
\end{center}
\caption{(i) Side view of the region $\mathcal{U}$ (shaded in grey). (ii) The region $U_p = C_p \cap \mathcal{U}$ (in red), together with the ball $B_p\subset U_p$ (in blue), indicating that $U_p$ is fat.}
\label{fig:th13-1}
\end{figure}
    
    We claim that each set $U_p$ contains a ball $B_p$ of radius $1/16$. Clearly the Voronoi cell $C_p$ contains this ball by the net property; it remains to show that $B_p$ is contained in $\cU$. Consider the line $\ell$ normal to $\cS$ at $p$, and let $p(\delta)$ be the point of this line on $\cS(\delta)$ where $\delta\in [-1/16,1/16]$. Notice that a point $q\in B_p$ the nearest point of $\ell$ is some point $p(\delta)$, and the tangent plane $H_\delta$ of $\cS(\delta)$ at $p(\delta)$ contains $q$. Since $|q,p(\delta)|\leq 1/16$, we have by \Cref{lem:between-offsets} that $q$ is between $\cS(\delta-2/16^2)$ and $\cS(\delta+2/16^2)$. Since $\delta<1/16$, this implies that $q\in \cU$, and concludes the proof that $B_p\subset U_p$.

    Consider now the obstacle set $\cT$, and notice that its triangles have positive pairwise distance from each other and also from $\cS(-1/2)$ and from $\cS(1/2)$. Let $\mu>0$ be some small number such that the Euclidean $\mu$-neighborhoods of the triangles of $\cT$ remain pairwise disjoint, and they remain subsets of $\cU$. For a triangle $T\in \cT$ let $x(T)$ denote the center of $T$ and let $T^\mu$ denote the set of points $q$ in $\Reals^3$ such that $\dist_{\Reals^3}(q,T)<\mu$ (i.e., $T^\mu$ is an open set). Here $\dist_{\Reals^3}(q,T)$ denotes the minimum Euclidean distance between $q$ and a point of $T$.
    
    Let $U_p^{-\mu}$ denote the set of points $q$ with $\dist_{\Reals^3}(q,\Reals^3\setminus U_p)\geq \mu$, that is, $U_p^{-\mu}$ is a closed subset of $U_p$ where a small neighborhood of the boundary of $U_p$ has been removed. Consequently, the sets $U_p^{-\mu}$ are pairwise disjoint.
    
    We group $\cT$ according to the location of their triangle centers. Let $\cT_p$ denote the set of triangles $T\in \cT$ where $x(T)\in U_p$; if $x(T)$ is on the shared boundary of several sets $(U_p)_{p \in N}$, then we assign $T$ to one of the corresponding classes arbitrarily. As a result $\{\cT_p\}_{p \in N}$ is a partition of $\cT$.   
    For each point $p\in N$ we define the following obstacle $W_p$ by adding all triangles of $\cT_p$ to $U_p^{-\mu}$, but subtracting $T^\mu$ for all $T\in \cT\setminus \cT_p$:
    \[W_p:= \left(U_p^{-\mu} \cup \left(\bigcup_{T\in \cT_p} T \right)\right) \setminus \left(\bigcup_{T\in \cT\setminus \cT_p} T^\mu \right).\]
    Refer to \Cref{fig:th13-2} for an illustration of an obstacle $W_p$. The sets $W_p$ are computed in polynomial time, and they are closed sets that are pairwise disjoint. Moreover, each set $W_p$ has diameter at most $\diam(W_p)\leq \diam(U_p)+s_T<2$ where $s_T<1/100$ is the side length of the triangles in $T$.
    Recall that each set $U_p$ contains a ball of radius $1/16$ centered at $p$.
    It follows that $W_p$ contains the ball of radius $1/16-\mu-s_T>1/32$ centered at $p$. We conclude that the sets $W_p$ are $1/64$-fat.

\begin{figure}
\begin{center}
\includegraphics[scale=0.945]{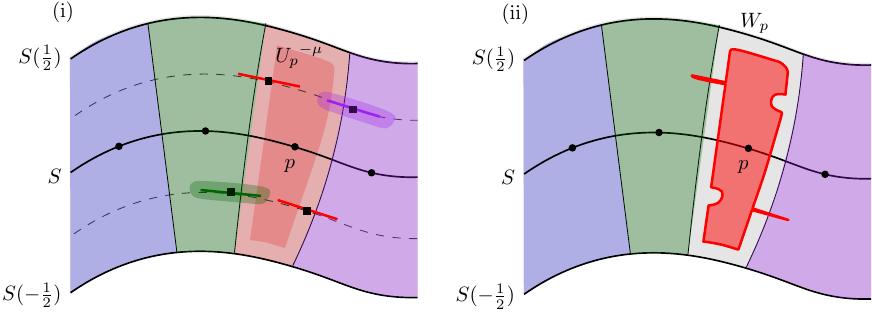}
\end{center}
\caption{(i) By shrinking $U_p$, we obtain the region $U_p^{-\mu}$. Square points correspond to centers of triangles in $\cT$. Four triangles are drawn, again in side view. For illustration purposes, various objects have not been drawn to scale. (ii) The final obstacle $W_p$, obtained by attaching to $U_p^{-\mu}$ the two red triangles, and removing a small neighborhood of the other two.}
\label{fig:th13-2}
\end{figure}

    The proof that the obstacle set $\{W_p\}_{p\in N}$ approximately realizes $(X,\dist_X)$ can be argued as in part (i). The number of obstacles is $|N|=O(n^5\Phi/\eps)$.
\end{proof}
We note that both the size of our obstacle collection and the construction time could be further improved; here, however, we did not attempt to optimize for these aspects.

\section{Algorithms and lower bounds for \otsp} \label{sec:algoandlower}

This section is devoted to proving \Cref{cor:algoandlower}. In order to give an algorithm, we will first need to (approximately) compute the distances between our input points. Let $\cT$ denote a collection of pairwise-disjoint polyhedral obstacles in $\Reals^3$ of total complexity $m$. Har-Peled \cite{Har-Peled99} showed how to construct an $\eps$-approximate shortest path map from any source point $s\in \cF_\cT$, in time roughly $O(m^4/\eps^6)$.  From his result we can deduce the following:

\begin{theorem}[Har-Peled~\cite{Har-Peled99}]\label{ptas-shortest-paths}
Let $\cT$ be a collection of pairwise-disjoint polyhedral obstacles in $\Reals^3$ of total complexity $m$ and let $S\subset \cF_\cT$, be a set of $|S|=n$ points. Then, for any $0<\eps<1$, we can compute a $(1+\eps)$-approximation for the all-pairs-shortest-paths problem on $S$ in time:

$$
O\left( n\frac{m^4}{\eps^2} \left( \frac{\beta(m)}{\eps^4} \log \frac{m}{\eps} + \log (m \rho) \log (m \log \rho) \right) \log \frac{1}{\eps} \right)
$$
where $\rho$ is the ratio of the longest edge in $\cT$ with the distance of the closest pair of points in $S$, and $\beta(m) = \alpha(m)^{O(\alpha(m))^{O(1)}}$, and $\alpha(m)$ is the inverse of the Ackermann function.
\end{theorem}

Note that Har-Peled's result is based on constructing a graph whose vertices correspond to points and whose edges correspond to pairwise visible point pairs, where edge lengths are high-precision estimates of the distance of the visible point pair. Har-Peled's result computes the true shortest paths in this graph. Consequently, the resulting approximate distance function $\dist_{\mathrm{approx}}$ on $S$ forms a metric space. One can also show using arguments similar to \Cref{le:doubling-space} that for any $\eps\in (0,1)$ the doubling dimension of $(S,\dist_{\mathrm{approx}})$ is at most constant times the doubling dimension of $(S,\dist_\cT)$.

 We now mention the result by Banerjee, Bartal, Gottlieb and Hovav, on TSP in doubling spaces \cite{banerjee24}.

\begin{theorem}[Corollary~35 in \cite{banerjee24}]\label{ptas-doubling}
    Let $(X,\dist_X)$ be a metric space with $|X|=n$ points and of bounded doubling dimension $\delta$. A $(1+\eps)$-approximation to the optimal tour of $X$ can be computed by a randomized algorithm in
    $2^{(\delta/\eps)^{O(\delta)}}n+(1/\eps)^{O(\delta)}n\log{n}$ time.
\end{theorem}

By combining Theorem~\ref{ptas-doubling} and Theorem~\ref{ptas-shortest-paths}, we will obtain a PTAS for the case $\mathfrak P = \text{ convex } \wedge  \text{ pairwise disjoint } \wedge \text{ fat}$. For the other cases, the APX-hardness is based on the  inapproximability result by Karpinski, Lampis and Schmied on \mtsp \cite{karpinski}. They showed an inapproximability bound of $123/122$. By inspection of the construction given in~\cite{karpinski}, we have the following crucial observation:

\begin{observation}\label{tsp-inapprox}
   It is NP-hard to approximate \mtsp within a ratio of $123/122$, even when the underlying metric space has polynomial spread.
\end{observation}

\begin{proof}
 This follows by inspection of the construction by Karpinski, Lampis and Schmied (Section~5.1 in \cite{karpinski}), which is based on a the shortest path metric of a connected edge-weighted graph. The weights of all edges used are small constants. Namely, the smallest edge has weight $0.5$, and the longest edge has length $2$. Therefore, the graph has diameter at most $2(n-1)$, and the construction has a spread at most $\frac{2(n-1)}{0.5} = 4(n-1)= O(n)$.   
\end{proof}

 Now we have all the necessary elements to prove \Cref{cor:algoandlower}.

\corAlgoLower*

\begin{proof}[Proof of \Cref{cor:algoandlower}]
    Let $S\subset \cF_\cT$ denote a set of $n$ points, whose optimal tour we want to compute. We have the following cases: 
    \begin{itemize}
        \item \textbf{The obstacles in $\cT$ are convex, $\alpha$-fat and pairwise disjoint.}
        
        We show that then \otsp admits a PTAS. Towards that, let $\eps\in (0,1)$. By Theorem~\ref{ptas-shortest-paths}, we can approximate $\dist_T$ up to a factor of $1+\eps/3$. Let $\dist_{\mathrm{approx}}$ denote the approximate metric. The doubling dimension of $(S,\dist_{\mathrm{approx}})$ is at most a constant factor larger than that of $(S,\dist_\cT)$. Therefore, by Lemma~\ref{le:doubling-space}, it is at most $O(1)$.
        By Theorem~\ref{ptas-doubling}, we can get a $(1+\eps/3)$-approximation for \otsp on $(S,\dist_{\mathrm{approx}})$, with running time $2^{1/\eps^{O(1)}}n+1/\eps^{O(1)}n\log{n}$. Since $(1+\eps/3)^2<1+\eps$, the result is a $(1+\eps)$-approximation for \otsp. 
        
    
    \item \textbf{The obstacles in $\cT$ are convex and pairwise disjoint but not fat, or fat and convex but non-disjoint, or fat and disjoint but non-convex.}
    
    We will show that in this case the problem is APX-hard, with an inapproximability bound of $\frac{123}{122}-\eps$, for any $\eps>0$. 
    For a set $S\subset \cF_\cT$, let $\textrm{OPT}_\cT(S)$ denote the length of an optimal tour through $S$. Assume that we could, in polynomial time, compute a tour of length $T_{\mathrm{approx}}(S)$ such that $T_{\mathrm{approx}}(S) \leq \left(\frac{123}{122}-\eps\right)\textrm{OPT}_\cT(S)$, for a fixed $\eps>0$.

    Let $(X,\dist_X)$ denote an instance of \mtsp, where $|X|=n$ and $X$ has spread $\textrm{poly}(n)$. From \Cref{tsp-inapprox}, we have a $123/122$ inapproximability bound for computing the optimal tour on $X$, denoted by $\textrm{OPT}(X)$. From~\Cref{thm:approxrealizemetric} or \Cref{thm:fatobstacles}, we compute a set of obstacles $\cT$ in $\mathbb{R}^3$ and an injection $f:X\rightarrow \Reals^3$, such that: 

    $$\textrm{OPT}_\cT(f(X))\leq (1+\eps') \textrm{OPT}(X),$$

    where we set $\eps'$ to satisfy $\frac{123}{122(1+\eps')}=\frac{123}{122}-\eps$.

    By our assumption, we can compute a value $T_{\mathrm{approx}}(f(X))$ such that:

    $$T_{\mathrm{approx}}(f(X))\leq \left(\frac{123}{122}-\eps\right)\textrm{OPT}_\cT(f(X))$$

    Therefore, $T_{\mathrm{approx}}(f(X))\leq \frac{123}{122} \textrm{OPT}(X)$, which contradicts the APX-hardness of \mtsp.\qedhere
    \end{itemize}
\end{proof}

\section{Conclusion}\label{sec:conclusion}

In this paper we have shown that all metric spaces are approximately realizable in $\Reals^3$ using convex disjoint obstacles, or using convex fat but non-disjoint obstacles, or using fat disjoint obstacles that are non-convex. Our convex construction is easily generalizable to several other obstacle types, however, it always requires objects whose minimum bounding box has two long edges (polynomially long in $n,1/\eps$ and the spread) and one short edge of length $1$.

Some notable obstacle types where the possible realizable metrics remain unexplored include
disjoint axis-parallel box obstacles and disjoint $1\times 1 \times n$ boxes of arbitrary orientation. We suspect that obstacles of the former type cannot realize all metric spaces, and obstacles of the latter type can only realize spaces of small doubling dimension. 
It also remains open whether exact realization of any metric space is possible with convex obstacles, even if all distances in the metric space are either $1$ or $2$. 

Finally, can we generalize the planar studies of distances among weighted regions~\cite{papad91,bose22,mitchell_wrp} to $\Reals^3$? What metrics can be realized if we are allowed to pass through the objects of $\cT$, but we need to pay some custom penalty (weight) for doing so?

\bibliography{ref}

\end{document}